\documentclass[reqno,11pt]{amsart}
\usepackage{amsmath,amsthm,amssymb,enumerate}
\usepackage[dvipsnames]{xcolor}
\usepackage{combelow}
\usepackage{bbm}
\usepackage{mathtools}
\usepackage{graphicx}
\usepackage{tikz}
\usepackage{tikz-cd}
\usetikzlibrary{calc}

\usepackage{amsfonts}
\usepackage{amsxtra}
\usepackage{mathrsfs}
\usepackage[left=4cm,right=4cm,top=3cm,bottom=2.5cm]{geometry}
\usepackage[colorlinks=true, linktocpage=true, linkcolor=red!70!black, citecolor=green!50!black]{hyperref}
\allowdisplaybreaks
\usepackage{enumitem}
\usepackage[mathscr]{eucal}
\setenumerate{label={\rm (\alph{*})}}
\numberwithin{equation}{section}
\allowdisplaybreaks[1]
\definecolor{labelkey}{gray}{.65}

\makeatletter
\newcommand*\bigcdot{\mathpalette\bigcdot@{.5}}
\newcommand*\bigcdot@[2]{\mathbin{\vcenter{\hbox{\scalebox{#2}{$\m@th#1\bullet$}}}}}
\makeatother
\newcommand{\Thanks}{\vspace*{.5em} \noindent \thanks}

\newcommand{\R}{\mathbb R}
\newcommand{\N}{\mathbb N}

\newcommand{\F}{\mathscr F}

\renewcommand{\leq}{\leqslant}

\DeclareMathOperator{\tr}{tr}

\newcommand{\dif}{\mathrm{d}}

\DeclareMathOperator{\supp}{supp}

\renewcommand{\dif}{\operatorname{d}\!}

\newcommand{\locc}{\operatorname{loc}}

\newcommand{\hold}{\operatorname{C}}

\renewcommand{\L}{\mathbb{L}}
\usepackage{esint}

\newcommand{\C}{\mathbb{C}}

\newcommand{\Lin}{\operatorname{L}}

\newcommand{\meas}{\mathfrak{M}}

\newcommand{\wasser}{W_{p}}
\renewcommand{\H}{\mathscr{H}}
\newcommand{\la}{\langle}
\newcommand{\ra}{\rangle}
\renewcommand{\L}{{\mathcal{L}}}
\newcommand{\Sact}{{\mathcal{S}}}
\newcommand{\K}{{\mathscr{K}}}
\newcommand{\m}{\mathfrak{m}}
\newcommand{\beq}{\begin{equation}}
\newcommand{\eeq}{\end{equation}}
\newcommand{\bitem}{\begin{itemize}[leftmargin=2.5em]}
\newcommand{\eitem}{\end{itemize}}

\normalsize

\allowdisplaybreaks

\numberwithin{equation}{section}

\newtheorem{theorem}{Theorem}[section]
\newtheorem{lemma}[theorem]{Lemma}
\newtheorem{proposition}[theorem]{Proposition}

\newtheorem{definition}[theorem]{Definition}

\newcommand{\QEDrem}{\ \hfill $\Diamond$}

\newtheorem{remark}[theorem]{Remark}

\numberwithin{equation}{section}

\begin{document}


\title[Action-Driven Flows for Causal Variational Principles]{Action-Driven Flows for \\ Causal Variational Principles}
\author[F.\ Finster]{Felix Finster}
\address[F.\ Finster]{Fakult\"at f\"ur Mathematik \\ Universit\"at Regensburg \\ D-93040 Regensburg \\ Germany}
\email{finster@ur.de}
\author[F.\ Gmeineder]{Franz Gmeineder \\ \\ March 2025 / March 2026} %
\address[F.\ Gmeineder]{Universit\"{a}t Konstanz, Fachbereich Mathematik \& Statistik, Universit\"{a}tsstrasse 10, D-78464 Konstanz, Germany }
\email{franz.gmeineder@uni-konstanz.de}

\maketitle

\begin{abstract}
We introduce action-driven flows for causal variational principles, being a class of non-convex variational problems emanating from applications in fundamental physics. In the compact setting, H\"older continuous curves of measures are constructed by using the method of minimizing movements. As is illustrated in examples, these curves will in general not have a limit point, due to the non-convexity of the action. This leads us to introducing a novel penalization which ensures the existence of a limit point, giving rise to approximative solutions of the Euler-Lagrange equations. The methods and results are adapted and generalized to the causal action principle in the finite-dimensional case. As an application, we construct a flow of measures for causal fermion systems in the infinite-dimensional situation.
\end{abstract}

\setcounter{tocdepth}{2}
\tableofcontents

\section{Introduction} \label{secintro}
The theory of {\em{causal fermion systems}} is a recent approach to fundamental physics
(for an introduction to the physical background and applications as well as
the mathematical context, we refer the interested reader to the
review~\cite{review}, the textbooks~\cite{cfs, intro} or the 
website~\cite{cfsweblink}).
In this approach, spacetime and all structures therein are encoded in a
measure~$\rho$ on a set of operators on a Hilbert space.
The physical equations are formulated via a variational
principle for the measure~$\rho$, the so-called causal action principle.
{\em{Causal variational principles}} evolved as a mathematical generalization
of the causal action principle~\cite{continuum, jet, noncompact}
(an introduction to the causal action principle and causal variational principles
can be found for example in~\cite[Chapters~5 and~6]{intro}).
From the point of view of the calculus of variations, causal variational principles
are a class of nonlinear, non-convex variational principles where one minimizes
an action~$\Sact$ under variations of a measure~$\rho$.
One of the objectives of the present paper is to formulate and analyze
corresponding {\em{flows of measures}}. Moving from the study of minimizing measures
to flows of measures can be understood in
analogy to the transition from stationary problems (like for example minimizing the
Dirichlet energy) to corresponding evolution equations (like for example the heat flow).
In simple terms, our flows can be understood as gradient flows corresponding to
causal variational principles.
Due to the lack of convexity and smoothness, the formulation of the flow equations
as well as the proof of existence of solutions are mathematically challenging
and seem of general interest in the context of non-smooth and
non-convex variational problems. 

\subsection{Causal variational principles} \label{seccvpintro}
In order to describe this objective and underlying obstructions in more detail,
we begin by recalling the general setting of causal variational principles. For simplicity, we firstly restrict attention to the so-called \emph{compact setting}; the detailed set-up shall be deferred to Section~\ref{secrprelimcvp} below.
	
Our starting point is a compact metric space~$(\F, d)$
and a non-negative function~${\L} : \F \times \F \rightarrow \R^+_0 := [0, \infty)$ (the {\em{Lagrangian}}) which is assumed to be continuous.
The corresponding \emph{causal action principle} then is to 
\begin{align}\label{eq:CAP}
	\text{minimize}\quad \Sact (\rho) = \int_\F \dif \rho(x) \int_\F \dif \rho(y)\: \L(x,y) 
\end{align}
over the class~$\meas_1(\F)$ of normalized Borel measures on~$\F$.
Causal variational principles are a class of examples for {\em{non-smooth}} and {\em{non-convex}} variational principles. The existence of solutions of \eqref{eq:CAP} is a consequence of the direct method of the Calculus of Variations
(see Section~\ref{secrprelimcvp}). Most importantly, minimizers~$\rho$ satisfy the corresponding {\em{Euler-Lagrange equations}} (\emph{EL equations} for brevity), and their precise formulation is given in  Section~\ref{secrprelimcvp}.

Constructing solutions of the EL equations -- or physically meaningful approximations thereof -- is of central importance in the theory of causal fermion systems
in order to get a better understanding of
the nature of the physical interactions as described by the causal action principle.
Here, abstract existence results are not sufficient, but one needs constructive methods
which give insight into the structure of the minimizing measure. By the aforementioned lack of smoothness and non-convexity, this is a non-trivial task in itself. In this regard, a central  objective  of the general theory is to find a canonical way of how a generic probability measure $\rho_{0}$ can be modified continuously to yield an (approximative) solution of the EL equations. In other words, this corresponds to a meaningful evolution $t\mapsto \varrho(t)$ with $\varrho(0)=\rho_{0}$ such that, for $t\to\infty$, $\varrho(t)$ approaches an (approximative) solution of the EL equations. 

\subsection{Gradient flows} By the variational nature of the problems considered here, it is natural to consider evolutions driven by the energies or actions given by \eqref{eq:CAP}. By this we mean that the energies of the solutions are decreasing in time.  Heuristically, this can be interpreted as a measure-valued variant of the ordinary differential equation 
\begin{align}\label{eq:ODEintro}
\left\{ \begin{array}{rll}
\displaystyle \frac{\dif}{\dif t} \varrho(t)\!\!\! &= - \nabla \Sact(\varrho) &\; \text{if~$t>0$}\:, \\[0.5em] 
 \varrho(0) \!\!\!&= \rho_{0}\:. 
 \end{array} \right.
\ \end{align}
However, for future reference, we remark that \eqref{eq:ODEintro} has to be understood symbolically; in our case and as shall be discussed below, this is due to the lack of smoothness, in turn being a consequence of the non-convexity and non-smoothness of the action $\mathcal{S}$. 

By way of comparison, in the more familiar situation of classical Dirichlet energies e.g.\ on Sobolev spaces, \eqref{eq:ODEintro} reduces to the usual heat equation. The convexity of the underlying energies then allows for useful a priori estimates, finally leading to both existence and regularity assertions for the respective flows. These methods have been refined
and extended to many other flow equations, provided that the driving energies are convex.

\subsection{Flows for non-convex variational problems}
The situation changes drastically if the underlying energies are no longer convex. To the best of our knowledge, there is no unifying theory that yields both existence and decisive statements on the long-time behavior of solutions of the associated gradient flows (see however related results in~\cite{rossi-savare, bellettini-novaga-paolini, rossi-segatti-stefanelli, muratori-savare, streets}).
To overcome the first issue, we employ a version of \textsc{De Giorgi}'s {\em{minimizing movements approach}}~\cite{degiorgi-minmov, ambrosio-minmov, braides-minmov, fleissner} adapted to the present setting; in essence, they can be understood as a method for extending the gradient flow to
non-smooth actions on infinite-dimensional spaces. This construction leads to a flow
\[ \Phi : [0,\infty)\times\meas_1(\F) \to \meas_1(\F) \]
with the property that the action given by \eqref{eq:CAP} is strictly decreasing along the flow lines. In essence, this is achieved by solving variational problems in discrete time steps which are penalized by the Wasserstein metric, and then pass to a continuous time evolution by use of an Arzel\`{a}-Ascoli-type argument. While we describe an analogous penalization procedure by use of the total variation norm, the use of the Wasserstein metric is most suitable here. Indeed, it is the weak*-convergence of probability measures for which compactness can be achieved and the actions \eqref{eq:CAP} are lower semicontinuous; the Wasserstein metric, in turn, induces weak*-convergence.
We also study the analogous procedure for the total variation norm. In this case,
we also get existence of a flow. But the flow has the shortcoming that it potentially
gets stuck away from local minima (as will be explained in an example in Section~\ref{secfurtherex}).
With this in mind, it seems that the Wasserstein distance is the correct metric
for the flow of measures we have in mind. We prove that the resulting curves of
measures are H\"older continuous (see Section~\ref{secflow}).

It is an important task to control the {\em{long-time behavior}} of solutions.
It is here where the interplay of non-convexity and the weak compactness properties of weak*-convergence necessitate additional arguments.
First, it is clear from the arbitrariness of the initial value $\rho_{0}$ that, at best, the curve will converge to an extremal point but not necessarily to a minimizer.
In fact, by the very definition of the flow, it might get stuck at a critical point of the functional, and by the non-convexity, the latter might be far away from any global minimizer.
In the general situation considered here, the situation is even worse: it may happen
that the gradient flow does not converge at all.
This will be shown in Section~\ref{secex} in a simple example where the potential
is constructed as a downward spiral with increasingly small potential wells
(see Figure~\ref{figexample} on page~\pageref{figexample}).
In examples of this type, which may be known to the experts in different scenarios,
there is not even a subsequence of times~$(t_{k})$ 
for which the measures converge to a solution of the EL equations.

In order to overcome such difficulties,
we also introduce another flow which involves an additional penalization 
term involving a parameter~$\xi>0$. In the case~$\xi=0$, we get back the above
flow by minimizing movements. In the case~$\xi>0$, the additional penalization
term gives us a-priori control
of the length of the curve (as measured in the Wasserstein distance) in terms of the
change of the action (see Section~\ref{seclip}).
This makes it possible to reparametrize the curve, using the action itself as the new parameter. In this way, we can circumvent the difficulty that the flow might get
stuck in ``plateaus'' of the potential for a long time (as shown in Figure~\ref{fig:reparamterise}
on page~\pageref{fig:reparamterise}).
After the reparametrization, the curve becomes even Lipschitz continuous (see Section~\ref{seclip}).
Moreover, we get control of the long-time behavior of the solutions.
Indeed, in the case~$\xi>0$ we prove that the
resulting curve~$\varrho^\xi(t)$ does converge (see Section~\ref{seclimit}).
The prize to pay is that the limiting measure satisfies the EL equations only approximately.
For the error term, we derive a precise a-priori bound which tends to zero as~$\xi \searrow 0$.
With this in mind, our procedure seems well-suited for the applications in mind.
For example, in a numerical study one can choose~$\xi$ so small that the error of the
approximation is bounded by the numerical errors.
%


We also extend our methods and results to the causal action principle for causal variational
principles. Our methods and results can be understood more generally from the perspective
of {{non-convex variational problems}}.
Indeed, causal variational principles are model examples of
variational principles which, in general, are fully non-convex. The methods to be developed in the present paper provide H\"older continuous flows of measures with these desired properties.

\subsection{Structure of the paper}
The paper is organized as follows. After the necessary preliminaries on
causal variational principles and measure theory (Section~\ref{secprelim}), 
we discuss a simple example of a non-smooth and non-convex variational problem in two dimensions
(Section~\ref{secex}). In Section~\ref{seccvp} flows are developed starting from minimizing movements
for causal variational principles in the compact setting.
In Section~\ref{secfurtherex} our results are illustrated by further examples.
Section~\ref{seccfs} is devoted to the
adaptation and generalization of our methods and results to the causal action principle in finite dimensions;
this section also includes a brief but self-contained introduction to causal fermion systems and the causal action principle.
Finally, in Section~\ref{secoutlook} we give an outlook on how our flow could be used for the study of the
EL equations for causal fermion systems in infinite dimensions.

\section{Preliminaries} \label{secprelim}
\subsection{Causal variational principles in the compact setting} \label{secrprelimcvp}
We let~$(\F, d)$ be a compact metric space
and suppose that the Lagrangian~$\L\colon\F\times\F\to\R_{0}^{+}$ satisfies the following assumptions:
\label{enumAB}
\begin{enumerate}[label=(A\arabic*)]
\item\label{item:ass1} $\L$ is {\em{symmetric}}: $\L(x,y)=\L(y,x)$ for all~$x,y\in\F$. 
\item\label{item:ass2} $\L \in \hold^0(\F \times \F, \R^+_0)$ is {\em{continuous}} in both arguments.
\end{enumerate}
The {\em{causal variational principle}} is to minimize the {\em{action}}~$\Sact$ defined as
the double integral over the Lagrangian
\beq \label{Sactdef}
\Sact (\rho) = \int_\F \dif \rho(x) \int_\F \dif \rho(y)\: \L(x,y)
\eeq
under variations of the measure~$\rho$ within the class of regular Borel measures,
keeping the total volume~$\rho(\F)$ fixed ({\em{volume constraint}}).
By rescaling the measure, it is no loss of generality to consider normalized measures, i.e., 
\[ \rho(\F) = 1 \:. \]
The existence of minimizers follows from standard compactness arguments
(see~\cite{continuum} or, in a slightly more general scenario, \cite[Section~3.2]{noncompact}
or~\cite[Chapter~12]{intro}); the method will also be revisited in Lemma~\ref{lem:exist1} below.

Given a minimizing measure~$\rho\in\meas_1(\F)$, we
introduce the underlying {\em{spacetime}}~$M$ as its support,
\[ M := \supp \rho := \F\setminus\bigcup \big\{ U\subset \F\;\text{open}\colon\;\rho(U)=0 \big\} \:. \]
In~\cite[Lemma~2.3]{jet} it was shown that a minimizer
satisfies the {\em{Euler-Lagrange (EL) equations}}, 
which state that the continuous function~$\ell : \F \rightarrow \R_0^+$ defined by
\[ 
\ell(x) := \int_\F \L(x,y)\: \dif \rho(y) \]
is minimal on spacetime,
\beq \label{EL}
\ell|_M \equiv \inf_\F \ell \:.
\eeq
For further details we refer to~\cite[Section~2]{jet} or~\cite[Chapter~7]{intro}; we remark that we left out the parameter~$\mathfrak{s}$ appearing in these contributions, which will not be required here.

\subsection{Background facts from optimal transport and metric measure spaces} \label{sec:opttrans}
We now fix our notation and recall a few background facts from measure theory and metric measure spaces
to be used in the sequel. We specialize the setting by assuming that~$\F$ is a compact {\em{metric}}
space with metric~$d$. 
We denote the set of probability measures on~$\F$ by~$\meas_{1}(\F)$.
More generally, we use~$\meas(\F)$ to denote the signed Radon measures on~$\F$ and endow~$\meas(\F)$ with the {\em{total variation norm}}
\beq \label{tvn}
\|\mu\|_{\meas(\F)} :=  \sup_{\pi\in\Pi(\F)}\sum_{B\in\pi}|\mu(B)|, \qquad \mu\in\meas(\F) \:,
\eeq
where~$\Pi(\F)$ is the set of all countable Borel partitions of~$\F$. For future reference, we note that ~$(\meas(\F), \|\mu\|_{\meas(\F)})$ is a Banach space, and that the metric induced by~$\|\cdot\|_{\meas(\F)}$, denoted by~$d_{\meas(\F)}$, will also referred to as the {\em{Fr{\'e}chet metric}}.

In our arguments below, we will also make use of the~$p$-{\em{Wasserstein metric}}
on~$\F$ for~$1 \leq p <\infty$.
Given a measure~$\mathbb{P}\in\meas_{1}(\F \times \F)$, for~$i\in\{1,2\}$ we denote
the projection to the~$i^\text{th}$ component by~$\pi^{i}\colon \F \times  \F \ni (x_{1},x_{2})\mapsto x_{i}\in \F$. We let~$\pi_{\#}^{i}\mathbb{P}(A):=\mathbb{P}(\pi_{i}^{-1}(A))$ for~$A \subset \F$ be the corresponding push-forward of~$\mathbb{P}$. As is customary in this context, we then define for~$\mu_{1},\mu_{2}\in\meas_{1}(\F)$ the class of {\em{couplings}}~$\Gamma(\mu_{1},\mu_{2})$ (also referred to as {\em{transport plans}}) by 
\begin{align*}
\Gamma(\mu_{1},\mu_{2}) :=\{\mathbb{P}\in\meas_{1}(\F\times\F)\colon\; \pi_{\#}^{i}\mathbb{P}=\mu_{i}\;\text{for}\;i\in\{1,2\}\} \:. 
\end{align*}
Here the measures~$\pi_{\#}^{i}\mathbb{P}$ are referred to as {\em{marginals}}.
Let~$1\leq p<\infty$. We then define for~$\mu,\nu\in\meas_{1}(\F)$ the~$p$\emph{-th Wasserstein metric} by
\beq \label{eq:WassersteinDef}
\wasser(\mu,\nu) := \bigg( \inf\Big\{\int_{\F\times\F} d(x,y)^{p}\:\dif\mathbb{P}(x,y)\;\colon 
\;\mathbb{P}\in\Gamma(\mu,\nu)\Big\} \bigg)^\frac{1}{p} \:. 
\eeq
The integral appearing in~\eqref{eq:WassersteinDef} will also be abbreviated by~$\mathbf{W}_{p}(\mathbb{P})$. For future reference, let us emphasize that~$W_{p}$ metrizes the weak*-convergence on~$\meas_{1}(\F)$, meaning that (see~\cite[Corollary~6.13]{villani})
\beq \label{eq:soundofwater1}
\Big(\int_{\F}\varphi\dif\mu_{j}\to \int_{\F}\varphi\dif\mu\;\;\;\text{for all}\;\varphi\in \hold(\F) \Big)
\qquad \Longleftrightarrow \qquad W_{p}(\mu_{j},\mu)\to 0,
\eeq
where~$\hold(\F)$ denotes the continuous functions on~$\F$.
The following lemma is clearly well-known to the experts,  but since it is crucial for our arguments below, we include its short proof. 

\begin{lemma}\label{lem:wasserbound}
For any~$p \in [1, \infty)$ the following inequality holds,
\beq \label{eq:soundofwater2}
W_p(\mu, \nu) \leq \text{\rm{diam}}(\F)\:\|\mu-\nu\|^{\frac{1}{p}}_{\meas(\F)} \qquad\text{for all}\;\mu,\nu\in\meas_{1}(\F) \:.
\eeq
Moreover, for any~$\mu,\nu\in\meas_{1}(\F)$ and~$\lambda\in [0,1]$,
\begin{align}\label{eq:wasserbound}
\wasser \big( \lambda\mu+(1-\lambda)\nu,\nu \big)\leq \lambda \wasser(\mu,\nu). 
\end{align}
\end{lemma} 
\begin{proof} For the proof of~\eqref{eq:soundofwater2} we introduce the measure
\[ \rho := \frac{1}{2}\: \big( \mu + \nu - |\mu-\nu| \big) \:. \]
Then the measures $\mu-\rho$ and $\nu - \rho$ are both positive, with total volume given by
\[ (\mu-\rho)(\F) = (\nu-\rho)(\F) = \frac{1}{2}\:\|\mu - \nu\|_{\meas(\F)} \:. \]
We consider the transport plan
\[ \mathbb{P}(x,y) := \rho(x) \: \delta(x,y) + \frac{2}{\|\mu - \nu\|_{\meas(\F)}}\:(\mu-\rho) \times (\nu-\rho) \:. \]
It has the desired marginals~$\pi_{\#}^1 \mathbb{P}= \mu$
and~$\pi_{\#}^2\mathbb{P}= \nu$. We thus obtain the estimate
\begin{align*}
W_p(\mu, \nu)^p &\leq \iint_{\F \times \F} d(x,y)^p \: d \mathbb{P}(x,y) \\
&\leq \text{diam}(\F)^p\:\frac{2}{\|\mu - \nu\|_{\meas(\F)}}\:(\mu-\rho)(\F)\: (\nu-\rho)(\F) \\
&= \frac{1}{2}\:\text{diam}(\F)^p\:\|\mu-\nu\|_{\meas(\F)} \:.
\end{align*}
This gives~\eqref{eq:soundofwater2}.

In order to prove~\eqref{eq:wasserbound}, we let~$\varepsilon>0$ be arbitrary and choose~$\mathbb{P}\in\Gamma(\mu,\nu)$, $\widetilde{\mathbb{P}}\in\Gamma(\nu,\nu)$ such that 
\begin{align}\label{eq:destillierteswasser}
\mathbf{W}_{p}(\mathbb{P}) < \wasser(\mu,\nu) + \varepsilon\quad\text{and}\quad 
\mathbf{W}_{p}(\widetilde{\mathbb{P}}) < \varepsilon.
\end{align}
Now it suffices to realize that the coupling~$\mathbb{P}':=\lambda\mathbb{P}+(1-\lambda)\widetilde{\mathbb{P}}$ has the two marginals
\[ \pi_{\#}^{1}\mathbb{P}'=\lambda\mu + (1-\lambda)\nu \qquad \text{and} \qquad \pi_{\#}^{2}\mathbb{P}'=\nu \:. \]
Hence~$\mathbb{P}'\in\Gamma(\lambda\mu+(1-\lambda)\nu,\nu)$ and therefore
\begin{align*}
\wasser(\lambda\mu+(1-\lambda)\nu,\nu)  \leq \mathbf{W}_{p}(\mathbb{P}') & = \lambda \mathbf{W}_{p}(\mathbb{P}) + (1-\lambda)\mathbf{W}_{p}(\widetilde{\mathbb{P}}) 
\stackrel{\eqref{eq:destillierteswasser}}{\leq} \lambda \wasser(\mu,\nu)+ \varepsilon. 
\end{align*}
Sending~$\varepsilon\searrow 0$ establishes~\eqref{eq:wasserbound}, and this completes the proof. 
\end{proof}

\section{An example of a non-smooth, non-convex variational principle} \label{secex}
In order to illustrate the familiar difficulties which one encounters when analyzing non-smooth,
non-convex variational principles, we begin with an explicit example.
Despite its simplicity, it has similar features as will be proven for general causal variational principles
later on. In order to keep the setting as simple as possible, instead of varying on a space of measures,
we consider a minimization problem for a function on~$\R^2$.
We choose polar coordinates~$(r, \varphi)$ and introduce the action~$\Sact$ by
\[ \Sact(r, \varphi) = \left\{ \begin{array}{cl} 
\displaystyle 3 - 2 r^2 + r^2 \,(1-r^2)\: \sin \Big( \frac{1}{1-r} + \varphi \Big) &\qquad \text{if~$r < 1$} \\
\exp(1-r) &\qquad \text{if~$r \geq 1$}\:.
\end{array} \right. \]
This action is smooth except on the unit circle~$r=1$, where it is merely continuous
(see the radial plot in Figure~\ref{figexample}).
\begin{figure}[t]
\begin{tikzpicture}
	\node at (0,0) {\includegraphics[scale=0.3]{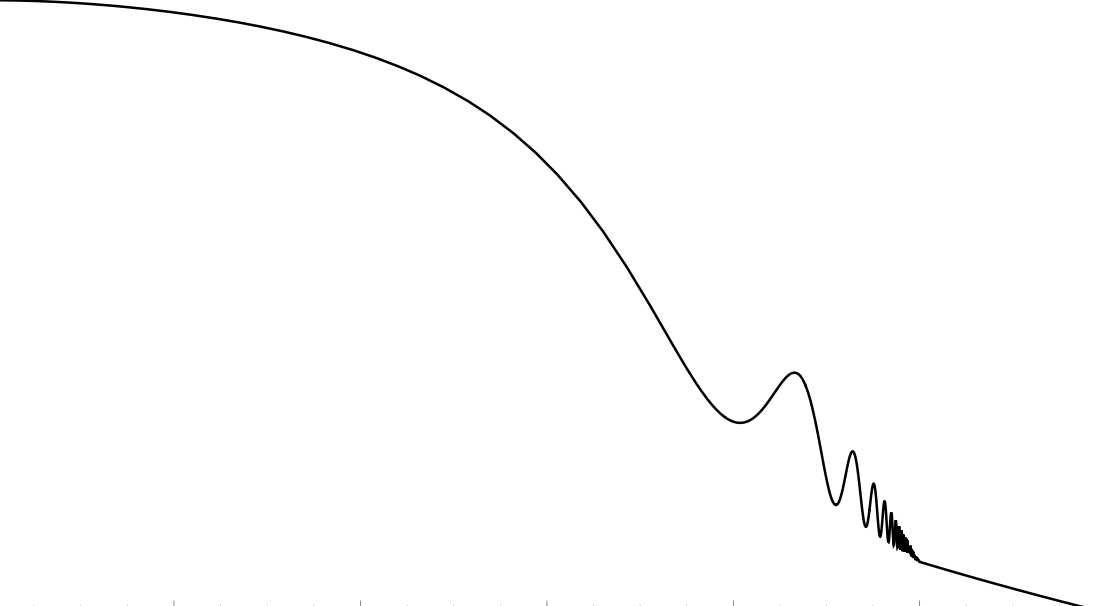}};
	\draw[->, thick] (-2.9,-1.7)--(-2.9,2);
	\draw[->, thick] (-3,-1.6) -- (3.5,-1.6);
	\node at (1,1) {$\mathcal{S}(r,0)$};
	\node at (3.7,-1.7) {$r$};
\end{tikzpicture}
\caption{Plot of the profile function~$\Sact(r,0)$.} \label{figexample}
\end{figure}
Suppose that we want to find a minimizer using a gradient flow, i.e.\
\beq \label{gradient}
\dot{\gamma}(t) = - \nabla \Sact|_{\gamma(t)} \qquad \text{and} \qquad
\gamma(0) = \Big(r=\frac{1}{5}, \varphi=0 \Big) \:.
\eeq
Then the curve~$\gamma(t)$ will ``spiral outward'' an infinite number of times.
Therefore, it will {\em{not}} converge,
\[ \lim_{t \rightarrow \infty} \gamma(t)  \qquad \text{does not exist}\:. \]
Instead, all the points of the unit circle are accumulation points of the curve.
However, the points on the unit circle itself are not critical, because the action becomes smaller
linearly if the radius is increased.
This gradient flow can be realized using minimizing movements if one considers the action
\beq \label{Spenal}
\Sact(r, \varphi) + \frac{1}{2h}\: d\big( (r,\phi), (r',\phi') \big)^2 \:,
\eeq
where~$d$ denotes the Euclidean distance in~$\R^2$.
Indeed, computing the first variation of this action in the Cartesian variables~$x=r \cos \phi$ and~$y=r \sin \phi$, we obtain the EL equations
\[ \begin{pmatrix} \partial_x \Sact \\ \partial_y \Sact \end{pmatrix} + \frac{1}{h} \: \begin{pmatrix} x-x' \\ y-y' \end{pmatrix} = 0 \:. \]
Assuming that the limit~$h \searrow 0$ exists, we obtain the differential equation~\eqref{gradient}. Therefore, the penalized action~\eqref{Spenal}
can be regarded as a discrete version of the gradient flow with step size~$h$.

We next consider minimizing movements with an additional penalization term
pa\-ra\-me\-trized by~$\xi>0$,
\[ \Sact(r, \varphi) + \frac{1}{2h}\: d\big( (r,\phi), (r',\phi') \big)^2 + \xi\, d\big( (r,\phi), (r',\phi') \big)\:. \]
Now the corresponding flow equation takes the form
\[ \dot{\gamma}_\xi(t) = 
\left\{ \begin{array}{cl} \displaystyle 
- \frac{\|\nabla \Sact|_{\gamma(t)}\|- \xi}{\|\nabla \Sact|_{\gamma(t)}\|} \; \nabla \Sact|_{\gamma(t)}
&\qquad \text{if~$\|\nabla \Sact|_{\gamma(t)}\| \geq \xi$} \\[1em]
0 &\qquad \text{otherwise}\:. \end{array} \right. \]
Therefore, the flow stops as soon as the norm of the gradient becomes smaller than~$\xi$.
Choosing~$\xi$ very small, the solution curve~$\gamma_\xi(t)$ will look similar to~$\gamma(\tau)$,
but instead of ``spiraling around'' an infinite number of times, it will stop at a point near the unit circle.
The resulting curve has finite length and a limit point,
\[ \gamma_{\xi}(\infty) := \lim_{t \rightarrow \infty} \gamma_\xi(t) \qquad \text{exists} \:. \]
The drawback is that the EL equations are satisfied only approximately in the sense that
\[ \big\| \nabla \Sact|_{\gamma_\xi(\infty)} \big\| \leq \xi \:. \]
In the limit~$\xi \searrow 0$, the limit points~$\gamma_\xi(\infty)$ again ``spiral around''
an infinite number of times. Therefore, the limit
\[ \lim_{\xi \searrow 0} \gamma_\xi(\infty) \qquad \text{does not exist}\:. \]
Instead, all the points of the unit circle are again accumulation points of the curve~$\gamma_\xi(\infty)$
with~$\xi \in \R^+$.

\section{Minimizing movements for causal variational principles} \label{seccvp}
\subsection{The causal action with penalization}
Throughout this section, we tacitly suppose that Assumptions~\ref{item:ass1} and~\ref{item:ass2} on the Lagrangian hold. In order to set up the minimizing movements scheme, we first consider variational problems with a given penalization. In particular, given parameters~$\xi \geq 0$, $h>0$ and a measure~$\rho$, we define
\begin{align}\label{eq:varprinB0}
\Sact^{h,\xi}(\mu):=\Sact(\mu)+\frac{1}{2h}\: d(\mu,\rho)^2+
\xi\: d(\mu,\rho) \:,
\end{align}
where~$d$ is the Fr{\'e}chet or the Wasserstein distance,
(cf.~\eqref{tvn} and~\eqref{eq:WassersteinDef})
\beq \label{dspecify}
\text{Case~1.} \;\; d=d_{\meas(\F)} \qquad \text{or} \qquad \text{Case~2.} \;\; d = W_p \:.
\eeq
The existence of solutions of the underlying minimization problem will be proven in Lemma~\ref{lem:exist1}. 
We begin with the following preparatory result (for a similar weaker statement see~\cite[Theorem~3.4]{noncompact}).
\begin{lemma}\label{lem:weakconv}
Let~$(\F, d)$ be a compact metric space and let~$\L\in\hold(\F\times\F)$. Then the functional 
\begin{align*}
\Sact \::\: \meas_{1}(\F)\ni\mu\mapsto \iint_{\F\times\F}\L(x,y)\dif\mu(x)\dif\mu(y)
\end{align*}
is continuous with respect to weak*-convergence on~$\meas_{1}(\F)$.

Moreover, the functional~$\Sact$ is Lipschitz continuous with respect to the Fr{\'e}chet metric, i.e.\
there is a constant~$C$ (which depends only on~$\F$ and~$\L$) such that
for all~$\rho, \tilde{\rho} \in \meas_{1}(\F)$,
\beq \label{lipes}
\big| \Sact(\tilde{\rho}) - \Sact(\rho) \big| \leq C\: d_{\meas(\F)}\big(\tilde{\rho}, \rho) \:.
\eeq
If we assume that the Lagrangian~$\L \in \hold^{0,\alpha}(\F \times \F, \R^+_0)$ is H\"older continuous with
H\"older exponent~$\alpha \in (0,1]$, then so is the functional~$\Sact$ with respect to the
Wasserstein distance, i.e.\ there is a constant~$C$ (which again depends only on~$\F$ and~$\L$) such that
for all~$\rho, \tilde{\rho} \in \meas_{1}(\F)$,
\beq \label{hoeles}
\big| \Sact(\tilde{\rho}) - \Sact(\rho) \big| \leq C\: W_p(\tilde{\rho}, \rho)^\alpha\:.
\eeq
\end{lemma}
\begin{proof} 
Let~$\rho,\rho_{1},\rho_{2},...\in\meas_1(\F)$ be such that~$\rho_{j}\stackrel{*}{\rightharpoonup}\rho$ as~$j\to\infty$.
Since~$\mathcal{F}$ is compact, the Weierstra\ss\, approximation theorem implies that the space
\[ X:=\mathrm{span}\{(x,y)\mapsto f(x)g(y)\colon\;f,g\in\hold(\F)\} \]
is dense in~$\hold(\F \times \F)$. Let~$\varepsilon>0$ be arbitrary but fixed. We then find 
$h\in \hold(X \times X)$ of the form $h(x,y)=\sum_{i=1}^{N}h_{i}f_{i}(x)g_{i}(y)$ with~$h_{1},...,h_{N}\in\R$ such that~$\|\L-h\|_{\infty}<\varepsilon$. Therefore, 
\begin{align*}
&\left\vert \iint_{\F\times\F}\L(x,y)\dif\rho(x)\dif\rho(y) - \iint_{\F\times\F}\L(x,y)\dif\rho_{j}(x)\dif\rho_{j}(y)\right\vert \\ 
& \leq \iint_{\F\times\F}|\L(x,y)-h(x,y)|\dif\rho(x)\dif\rho(y) \\
& \quad\:+  \left\vert
\iint_{\F\times\F}h(x,y) \dif\rho_{j}(x) \dif\rho_{j}(y) 
- \iint_{\F\times\F}h(x,y) \dif\rho(x) \dif\rho(y) \right\vert \\ 
& \quad\:+ \iint_{\F\times\F}|\L(x,y)-h(x,y)|\dif\rho_{j}(x)\dif\rho_{j}(y) =: \mathrm{I}+\mathrm{II}+\mathrm{III}. 
\end{align*}
We then have~$\mathrm{I} \leq \varepsilon\rho(\F)^{2}$ and~$\mathrm{III}\leq \varepsilon m^{2}$. On the other hand, by the very structure of~$h$, the weak*-convergence~$\rho_{j}\stackrel{*}{\rightharpoonup}\rho$ implies
\begin{align*}
\iint_{\F\times\F}h(x,y)\dif\rho_{j}(x)\dif\rho_{j}(y) & = \sum_{i=1}^{N}h_{i}\Big(\int_{\F}f(x)\dif\rho_{j}(x) \Big)\Big(\int_{\F}g(y)\dif\rho_{j}(y) \Big) \\ 
& \to \sum_{i=1}^{N}h_{i}\Big(\int_{\F}f(x)\dif\rho(x) \Big)\Big(\int_{\F}g(y)\dif\rho(y) \Big)\\ 
& = \iint_{\F\times\F}h(x,y)\dif\rho(x)\dif\rho(y) 
\end{align*}
as~$j\to\infty$, so that~$\mathrm{II}\to 0$ as~$j\to\infty$. By arbitrariness of~$\varepsilon>0$, the proof
of continuity is complete.

In order to prove the Lipschitz bound~\eqref{lipes}, we rewrite the difference of the actions as
\begin{align*}
&\Sact \big( \tilde{\rho} \big) - \Sact \big( \rho \big)
= \int_\F \dif \tilde{\rho}(x) \int_\F \dif \tilde{\rho}(y)\: \L(x,y) - 
\int_\F \dif \rho(x) \int_\F \dif \rho(y)\: \L(x,y) \\
&= \int_\F \dif \tilde{\rho}(x) \int_\F \dif \big( \tilde{\rho}- \rho \big)(y) \: \L(x,y)
+ \int_\F \dif \big(\tilde{\rho}- \rho\big)(x) \int_\F \dif \rho(y)\: \L(x,y) \:.
\end{align*}
Using that the Lagrangian is uniformly bounded and that the measures are normalized, we obtain
the estimate,
\[ \Sact \big( \tilde{\rho} \big) - \Sact \big( \rho \big)
\leq 2\,\|\L\|_{\hold^0(\F \times \F)}\: d_{\meas(\F)}\big(\tilde{\rho}, \rho) \:, \]
proving~\eqref{lipes}.

In order to derive the H\"older estimate~\eqref{hoeles}, we let~$\nu \in \meas_1(\F \times \F)$ be a coupling of~$\rho$ and~$\tilde{\rho}$. Then, using that the two marginals of~$\nu$ coincide with~$\rho$
and~$\tilde{\rho}$, the difference of actions can be written as
\[ \Sact(\tilde{\rho}) - \Sact(\rho) = \int_{\F \times \F} \dif \nu(x,x') \int_{\F \times \F} \dif \nu(y,y')
\big( \L(x',y') - \L(x,y) \big) \:. \]
Using that the Lagrangian is H\"older continuous with H\"older constant denoted by~$c$, we know that
\begin{align*} \big| \L(x',y') - \L(x,y) \big| & \leq \big| \L(x',y') -\L(x,y') \big| + \big| \L(x,y') - \L(x,y) \big| \\ &
\leq c\: \big( d(x,x')^\alpha + d(y,y')^\alpha \big) \:. 
\end{align*}
We thus obtain
\begin{align*}
\big| \Sact(\tilde{\rho}) - \Sact(\rho) \big| &\leq 2 c\: 
\int_{\F \times \F} d(x,x')^\alpha \: \dif \nu(x,x') \leq
2c \:\bigg( \int_{\F \times \F} d(x,x')^p \: \dif \nu(x,x') \bigg)^\frac{\alpha}{p},
\end{align*}
where in the last step we applied the H\"older inequality for normalized measures.
Taking the infimum over all couplings gives the result.
\end{proof}
 
\begin{lemma}\label{lem:exist1}
For any~$\xi\geq 0$, $h>0$ and~$\rho\in\meas_{1}(\F)$, there exists a minimizer~$\mu\in\meas_{1}(\F)$ of the causal action with penalization~\eqref{eq:varprinB0}.
\end{lemma} 
\begin{proof} 
Since~$\L\colon\F\times\F\to\R_{0}^{+}$, $\mathcal{S}^{h,\xi}$ is bounded below on~$\meas_{1}(\F)$ and thus~$m:=\inf_{\meas_{1}(\F)}\mathcal{S}^{h,\xi}$ exists in~$[0,\infty)$, we can choose a minimizing sequence~$(\mu_{j})\subset\meas_{1}(\F)$ for~$\mathcal{S}^{h,\xi}$, so that in particular~$m=\lim_{j\to\infty}\mathcal{S}^{h,\xi}(\mu_{j})$. By the duality relation~$\hold_{0}(\F)'\cong\meas(\F)$ and using that~$\meas_{1}(\F)$ is convex and closed, the Banach-Alaoglu theorem provides us with a non-relabeled subsequence and a probability measure~$\mu\in\meas_{1}(\F)$ such that we have~$\mu_{j}\stackrel{*}{\rightharpoonup}\mu$ in~$\meas_{1}(\F)$. 
By Lemma~\ref{lem:weakconv}, $\mathcal{S}$ is continuous with respect to weak*-convergence. Now, if (i) $d$ is the Fr\'{e}chet metric, then~$d(\cdot,\rho)=\|\cdot-\rho\|_{\mathfrak{M}(\F)}$ is lower semicontinuous with respect to weak*-convergence. On the other hand, if (ii) $d$ is the $p$-Wasserstein metric, then~$d$ metrizes weak*-convergence and so, in particular, $d(\cdot,\rho)$ is continuous with respect to weak*-convergence. In both cases, $\Sact^{h,\xi}$ is lower semicontinuous with respect to weak*-convergence. Hence, 
\begin{align*}
m \leq \Sact^{h,\xi}(\mu)\leq \liminf_{j\to\infty}\Sact^{h,\xi}(\mu_{j}) = m \:, 
\end{align*}
and therefore~$\mu$ is a minimizer.
\end{proof}
For clarity, we point out that minimizers will in general not be unique. Moreover, whereas the
Fr{\'e}chet metric~$d_{\meas(\F)}$ might seem as an easier or more natural choice, it comes with unfavorable properties of the flow (see Section~\ref{secfurtherex}) which can be avoided by working with the Wasserstein distance~$W_{p}$. 

\subsection{Minimizing movements}
Let~$\rho_{0}\in\mathfrak{M}_{1}(\mathcal{F})$ be a given initial measure.  Throughout, we fix a penalization parameter~$\xi \geq 0$ and, given~$h>0$, consider the sequence of
measures~$(\rho^{h, \xi}_j)_{j \in \N_0}$ obtained by choosing~$\rho^{h, \xi}_{j=0}=\rho_0$
and by iteratively minimizing the associated functional
\begin{align}\label{eq:varprinB}
\Sact_{j}^{h,\xi}(\mu):=\Sact(\mu)+\frac{1}{2h}\: d \big(\mu,\rho_{j-1}^{h,\xi} \big)^2+
\xi\: d(\mu,\rho_{j-1}^{h,\xi})
\end{align}
for~$j=1,2,\ldots$. The first penalization term follows the general procedure
in the minimizing movements approach (see for example~\cite{ambrosio-minmov});
also the resulting H\"older estimates (as in Lemma~\ref{lem:Holdreg}
and Proposition~\ref{thm:main1}) are adaptations of standard arguments to our setting
(see for example~\cite[Proposition~7.1]{braides-minmov}).
The second penalization term in~\eqref{eq:varprinB}, however, is novel.
The necessity of introducing this additional penalization term depending on~$\xi$ will be explained in detail in Section~\ref{seclip}. 

We begin by collecting several elementary estimates,
where~$d$ is again the distance function induced by either the Fr{\'e}chet metric or the Wasserstein distance~\eqref{dspecify}: 
\begin{lemma} The sequence~$(\rho^{h, \xi}_j)_{j \in \N_0}$ satisfies for all~$j \in \N$ the inequalities
\begin{align}
 \Sact(\rho_{j}^{h,\xi}) &\leq \Sact \big( \rho_{j-1}^{h,\xi} \big) \label{i0} \\
 d \big( \rho_{j}^{h,\xi},\rho_{j-1}^{h,\xi} \big) &\leq 
\frac{1}{\xi} \: \Big( \Sact \big( \rho_{j-1}^{h,\xi} \big) - \Sact \big( \rho_{j}^{h,\xi} \big) \Big) \label{i1} \\
d \big( \rho_{j}^{h,\xi},\rho_{j-1}^{h,\xi} \big) &\leq \sqrt{ 2 h\,
\big( \Sact(\rho_{j-1}^{h,\xi})-\Sact(\rho_{j}^{h,\xi}) \big) } \:. \label{i2}
\end{align}
Moreover, the inequality~\eqref{i0} is strict unless~$\rho_{j}^{h,\xi} = \rho_{j-1}^{h,\xi}$.
\end{lemma}
\begin{proof} The minimality implies that
\begin{align*} 
\Sact(\rho_{j}^{h,\xi})+\frac{1}{2h}d(\rho_{j}^{h,\xi},\rho_{j-1}^{h,\xi})^2 + \xi \:d(\rho_{j}^{h,\xi},\rho_{j-1}^{h,\xi}) & = \Sact_{j}^{h,\xi}(\rho_{j}^{h,\xi}) \\ &\leq \Sact_{j}^{h,\xi}(\rho_{j-1}^{h,\xi}) = \Sact(\rho_{j-1}^{h,\xi}) \:.
\end{align*}
Using that the terms on the left are all non-negative, the result follows immediately.
\end{proof}

\subsection{A H\"older continuous flow} \label{secflow}
Our goal is to show that, taking a suitable limit~$h \rightarrow 0$, we
to obtain a H\"older continuous curve~$\varrho^\xi(t)$ with~$t \in \R^+_0$.
In preparation, we form the continuous curve~$\rho^{h, \xi}$ by interpolation,
\beq \label{eq:curvedefineA}
\rho^{h, \xi}(t) := 
\bigg( \Big\lfloor \frac{t}{h}+1\Big\rfloor -\frac{t}{h} \bigg)\:\rho_{\lfloor\frac{t}{h}\rfloor}^{h,\xi} + 
\bigg( \frac{t}{h}
- \Big\lfloor \frac{t}{h} \Big\rfloor \bigg) \: \rho_{\lfloor\frac{t}{h}+1\rfloor}^{h,\xi} \:.
\eeq
For the next construction steps, we need the following generalization of the usual Arzel\`{a}-Ascoli theorem: 
\begin{lemma}[{\cite[Prop. 3.3.1]{ambrosio+gigli}}]\label{lem:ArzAsc}
Let~$(X,d)$ be a complete metric space and~$T>0$. Given a subset~$K\subset X$ which is sequentially compact with respect to a topology~$\tau$, suppose that~$(u_{j})_{j\in\mathbb{N}}$ is a sequence of maps~$u_{j}\colon [0,T]\to X$ such that 
\begin{gather}
u_{j}(t)\in K\qquad\text{for all}\;j\in\mathbb{N}\;\text{and all}\;t\in [0,T], \label{eq:ArzAscolA1} \\
\limsup_{j\to\infty} d \big( u_{j}(s),u_{j}(t) \big)\leq \omega(s,t)\qquad\text{for all}\;s,t\in [0,T] \:, \label{eq:ArzAscolA2}
\end{gather}
where~$\omega\colon [0,T]\times [0,T]\to [0,\infty)$ is a symmetric function (i.e.~$\omega(s,t)=\omega(t,s)$ for all~$s,t\in [0,T]$) with the property that~$\lim_{(s,t)\to(0,0)}\omega(s,t)=0$. Then there exists a subsequence~$(u_{j(k)})_{k\in\mathbb{N}}\subset (u_{j})_{j\in\mathbb{N}}$ and a $d$-continuous map~$u\colon [0,T]\to X$ such 
that the sequence~$(u_{j(k)})$ converges pointwise to~$u$ with respect to the topology~$\tau$.
\end{lemma}
Its applicability in the present framework follows from the following lemma:
\begin{lemma} \label{lem:Holdreg}
The curve~$\rho^{h, \xi}(t)$ defined by~\eqref{eq:curvedefineA} satisfies for all~$0<t_{1},t_{2}<\infty$
the inequality
\beq \label{eq:hoeldbound}
d \big( \rho^{h, \xi}(t_{1}),\rho^{h, \xi}(t_{2}) \big)\leq \sqrt{2}\: \sqrt{|t_{2}-t_{1}|+h}\;
\sqrt{\Sact(\rho_{0})} \:.
\eeq
\end{lemma} 
\begin{proof} It clearly suffices to consider the case~$t_{1}<t_{2}$. Then, by definition of~$\rho_{h,\xi}$,
\begin{align*}
&d \big( \rho^{h, \xi}(t_{1}),\rho^{h, \xi}(t_{2}) \big)  \leq d \Big(
\big( \lfloor\tfrac{t_{1}}{h}+1\rfloor -\tfrac{t_{1}}{h} \big)\: \rho_{\lfloor\frac{t_{1}}{h} \rfloor}^{h,\xi}
+ \big( \tfrac{t_{1}}{h}-\lfloor\tfrac{t_{1}}{h}\rfloor \big) \: \rho_{\lfloor\frac{t_{1}}{h}+1\rfloor}^{h,\xi},\;\rho_{\lfloor\frac{t_{1}}{h}+1\rfloor}^{h,\xi} \Big) \\ 
&\;\; + d \Big( \rho_{\lfloor\frac{t_{2}}{h}\rfloor}^{h,\xi},\;
\big( \lfloor\tfrac{t_{2}}{h}+1\rfloor -\tfrac{t_{2}}{h} \big)\: \rho_{\lfloor\frac{t_{2}}{h}\rfloor}^{h,\xi}
+ \big( \tfrac{t_{2}}{h}-\lfloor\tfrac{t_{2}}{h}\rfloor \big)\: \rho_{\lfloor\frac{t_{2}}{h}+1\rfloor}^{h,\xi} \Big) 
+ \!\!\!\!\sum_{j=\lfloor\frac{t_{1}}{h}+1\rfloor}^{\lfloor\frac{t_{2}}{h}\rfloor -1}\!\!\!
d\big( \rho_{j}^{h,\xi},\rho_{j+1}^{h,\xi} \big) \\
& \leq  \Big( \lfloor\tfrac{t_{1}}{h}+1\rfloor -\tfrac{t_{1}}{h} \Big)\:
d\big(\rho_{\lfloor\frac{t_{1}}{h}\rfloor}^{h,\xi},\: \rho_{\lfloor\frac{t_{1}}{h}+1\rfloor}^{h,\xi} \big)
+ \Big( \tfrac{t_{2}}{h}-\lfloor\tfrac{t_{2}}{h}\rfloor \Big)\: d \big( \rho_{\lfloor\frac{t_{2}}{h}\rfloor}^{h,\xi},\rho_{\lfloor\frac{t_{2}}{h}+1\rfloor}^{h,\xi} \big) \\
&\qquad \qquad
+ \sum_{j=\lfloor\frac{t_{1}}{h}+1\rfloor}^{\lfloor\frac{t_{2}}{h}\rfloor -1}
d \big( \rho_{j}^{h,\xi},\rho_{j+1}^{h,\xi} \big) \:,
\end{align*}
where the last step is trivial for~$d$ being the Fr\'{e}chet metric and follows from 
Lemma~\ref{lem:wasserbound} in the case of the Wasserstein metric.
It follows that
\begin{align*}
&d \big( \rho^{h, \xi}(t_{1}),\rho^{h, \xi}(t_{2}) \big) \leq  \sum_{j=\lfloor\frac{t_{1}}{h}\rfloor}^{\lfloor\frac{t_{2}}{h}+1\rfloor -1}d(\rho_{j}^{h,\xi},\rho_{j+1}^{h,\xi}) \\
&\!\!\overset{\eqref{i2}}{\leq} \sum_{j=\lfloor\frac{t_{1}}{h}\rfloor}^{\lfloor\frac{t_{2}}{h}+1\rfloor -1}
\sqrt{ 2h\: \big( \Sact(\rho_{j}^{h,\xi})-\Sact(\rho_{j+1}^{h,\xi}) \big) }\\ 
& \leq \bigg( \sum_{j=\lfloor\frac{t_{1}}{h}\rfloor}^{\lfloor\frac{t_{2}}{h}+1\rfloor -1}1\bigg)^{\frac{1}{2}}\bigg(\sum_{j=\lfloor\frac{t_{1}}{h}\rfloor}^{\lfloor\frac{t_{2}}{h}+1\rfloor -1} 2h \:\big(\Sact(\rho_{j}^{h,\xi})-\Sact(\rho_{j+1}^{h,\xi}) \big) \bigg)^{\frac{1}{2}} \:.
\end{align*}
The last sum is telescopic. Moreover, using that the sequence of actions
is monotone decreasing~\eqref{i0} and non-negative, we conclude that
\[ d \big( \rho^{h, \xi}(t_{1}),\rho^{h, \xi}(t_{2}) \big)
\leq \sqrt{2} \: \sqrt{ (t_{2}-t_{1})+h} \: \sqrt{\Sact(\rho_{0})} \:. \]
 This completes the proof.
\end{proof}

We are now ready for proving our first existence result.
\begin{proposition}\label{thm:main1}
For any~$\xi \geq 0$, there is a H\"older continuous flow
\[ \varrho^\xi \in\hold^{0,\frac{1}{2}} \big( [0,\infty), (\meas_1(\F),d) \big) \]
with~$\varrho(0)=\rho_{0}$. Setting
\[ 
t_{\max} := \inf \big\{ t \in \R^+ \:\big|\: \Sact\big( \varrho^\xi(t) \big) = \inf_{\tau \in \R^+}
\Sact \big( \varrho^\xi(\tau) \big) \big\} \;\in\; \R^+ \cup \{\infty\} \:, \]
the action is strictly monotone decreasing up to~$t_{\max}$, i.e.\
\beq \label{monotone}
\Sact\big( \varrho^\xi(t_1) \big) > \Sact\big( \varrho^\xi(t_2) \big) \qquad
\qquad \text{for all~$0 \leq t_1 < t_2 \leq t_{\max}$}\:.
\eeq
Moreover, the flow curve satisfies for all~$0 \leq t_1 < t_2$ the H\"older bound
\[ 
d \big( \varrho^{\xi}(t_{1}),\varrho^{\xi}(t_{2}) \big)\leq \sqrt{2}\: \sqrt{t_{2}-t_{1}}\;
\sqrt{\Sact(\rho_{0})} \:. \]
\end{proposition}
\begin{proof}
\emph{Case 1. $d=\wasser$.} Let~$[T_{1},T_{2}]\subset [0,\infty)$ be a compact interval. We note that~$(\meas_{1}(\F),\wasser)$ is a compact, hence complete, metric space by the Banach-Alaoglu theorem. We aim to apply Lemma~\ref{lem:ArzAsc} to the sequence~$(\rho^{\xi, 1/j})_{j\in\mathbb{N}}$
together with~$d=\wasser$ and~$\tau$ being the weak*-topology on~$\meas_{1}(\F)$. Then~$\rho^{\xi, 1/j}(t)\in K:=\meas_{1}(\F)$ for all~$j\in\mathbb{N}$, whereby~\eqref{eq:ArzAscolA1} is satisfied. 
Moreover, the estimate~\eqref{eq:hoeldbound} yields that~\eqref{eq:ArzAscolA2} is fulfilled with~$\omega(s,t):=\sqrt{2|s-t|}$. Consequently,
Lemma~\ref{lem:ArzAsc} together with~\eqref{eq:soundofwater1}
gives the existence of a $\wasser$-continuous limit map~$\rho^{\xi}\colon [T_{1},T_{2}]\to\meas_{1}(\F)$ such that~$\rho^{\xi, 1/j(k)}(t) \to \varrho^{\xi}(t)$ with respect to~$d=\wasser$ for every~$t \in [T_{1},T_{2}]$. For all~$T_{1}\leq t_{1}\leq t_{2}\leq T_{2}$ we thus obtain 
\begin{align*}
\wasser(\varrho^{\xi}(t_{1}),\varrho^{\xi}(t_{2})) & \leq \limsup_{k\to\infty}\Big(\wasser \big(
\varrho^{\xi}(t_{1}),\rho^{\xi, 1/j(k)}(t_{1}) \big)+\wasser \big( \rho^{\xi, 1/j(k)}(t_{1}),\rho^{\xi, 1/j(k)}(t_{2}) \big
)\big. \\
& \qquad\qquad\qquad +\wasser \big( \rho^{\xi, 1/j(k)}(t_{2}),\varrho^{\xi}(t_{2}) \big) \Big)
\leq \sqrt{2}\: \sqrt{t_{2}-t_{1}}\;
\sqrt{\Sact(\rho_{0})}
\end{align*}
by everywhere convergence and the estimate~\eqref{eq:hoeldbound}.

In order to construct the requisite curve as claimed in Proposition~\ref{thm:main1}, we cover~$[0,\infty)$ by intervals~$I_{\ell}:=[\ell-1,\ell+1]$, $\ell\in\mathbb{N}$. By what has been said above, we may choose a sequence~$(j_{k}^{(1)})$ such that, for a certain limit curve~$\varrho^{\xi}\in\hold^{0,1/2}([0,2];\meas_{1}(\F))$ we have 
\begin{align*}
\rho^{\xi, 1/j_k^{(1)}}\to\varrho^{\xi}
\end{align*}
with respect to~$\wasser$ on~$[0,2]$ as~$k\to\infty$. Next choose a subsequence~$(j_{k}^{(2)})\subset(j_{k}^{(1)})$ such that 
\begin{align*}
\rho^{\xi, 1/j_{k}^{(2)}} \to\overline{\varrho}^{\xi}
\end{align*}
for a certain limit curve~$\overline{\varrho}^{\xi}\in\hold^{0,1/2}([1,3];\meas_{1}(\F))$. Clearly, since~$(h_{k}^{2})\subset(h_{k}^{1})$, we must have~$\varrho=\overline{\varrho}$ on~$[1,2]$, and then define~$\varrho^{\xi}:=\overline{\varrho}^{\xi}$ on~$[2,3]$. Proceeding iteratively in this way and passing to the diagonal sequence, we obtain a sequence~$(j_{l})$ with~$j_{l}\to\infty$ and a curve~$\varrho^{\xi}\in\hold([0,\infty);(\meas_{1}(\F),\wasser))\cap\hold^{0,1/2}([0,\infty);(\meas_{1}(\F),\wasser))$ such that for any compact subset~$I\subset [0,\infty)$ there holds
\[ 
\rho^{\xi, 1/j_{l}}(t)\to \varrho(t)\;\;\text{for all~$t\in I$ in~$(\meas_{1}(\F),\wasser)$} \:. \]

\vspace*{0.5em}

\noindent
\emph{Case 2. $d=d_{\meas(\F)}$.} In this situation, we let~$d=d_{\meas(\F)}$ and again let~$\tau$ be the weak*-topology on~$\meas_{1}(\F)$. Then~$K:=\meas_{1}(\F)$ is compact for~$\tau$.  Arguing as above, specifically applying~\eqref{eq:hoeldbound} to~$d=d_{\meas(\F)}$, we obtain the existence of a limit map~$\varrho^{\xi}
\in\hold([0,\infty);(\meas_{1}(\F);d_{\meas(\F)}))$ such that, for some sequence~$(j_{l})$ with~$j_{l}\to\infty$ as~$l\to\infty$, $\rho^{\xi, 1/j_{l}}\to \varrho$
 in~$(\meas_{1}(\F),\wasser)$ (not in~$(\meas_{1}(\F),d_{\meas(\F)})$),
 locally uniformly in time (i.e.\ uniformly in~$t$ in a compact subset of~$[0, \infty)$).

Let us note that we have~$\varrho^{\xi}\in\hold_{\locc}^{0,1/2}([0,\infty);(\meas_{1}(\F),d_{\meas(\F)}))$ indeed: Let~$0\leq T_{1}\leq T_{2}<\infty$, so that~$\rho^{\xi, 1/j_{l}}(t)\stackrel{*}{\rightharpoonup}\varrho^{\xi}(t)$ for all~$t\in [T_{1},T_{2}]$ since~$\wasser$ metrizes weak*-convergence on~$\meas(\F)$. Since in the present setting~\eqref{eq:hoeldbound} is available for~$d=d_{\meas(\F)}$, 
we conclude for~$t, t' \in [T_{1},T_{2}]$ by weak*-lower semicontinuity of the total variation norm
\begin{align*} 
\|\varrho^{\xi}(t)-\varrho^{\xi}(t')\|_{\meas(\F)} & \leq \liminf_{l\to\infty} \|\varrho_{1/j_{l}}^{\xi}(t)-\varrho_{1/j_{l}}^{\xi}(t') \|_{\meas(\F)} 
\stackrel{\eqref{eq:hoeldbound}}{\leq} \sqrt{2 \Sact(\rho_0)}\: |t-t'|^{\frac{1}{2}} \:. 
\end{align*}
In this sense, the passage to the weak*-metric is only required to obtain the existence of such a curve, whereas the H\"{o}lder regularity for~$d_{\meas(\F)}$ survives from Lemma~\ref{lem:Holdreg} by lower semicontinuity. This concludes the proof of Proposition~\ref{thm:main1}.
\end{proof}

Note that the previous theorem holds both in the case of the Wasserstein metric and the Fr{\'e}chet
metric on~$\meas_{1}(\F)$. However, the flow in these two cases has quite different properties,
as will be illustrated in Section~\ref{secfurtherex} by a few examples.

\subsection{A Lipschitz curve in the case~$\xi>0$} \label{seclip}
The introduction of a positive penalization parameter~$\xi>0$ in \eqref{eq:varprinB} is motivated by the fact that it gives us curves of finite length in~$\meas(\F)$. In order to see this, we iterate ~\eqref{i1} and use
again that~${\Sact}$ is monotone decreasing. We thus obtain
\beq \label{totvarest}
\sum_{j={n+1}}^N d \big( \rho_j^{h,\xi},\rho_{j-1}^{h,\xi} \big) \\
\leq \frac{1}{\xi} \:\big( \Sact(\rho_n^{h,\xi})-\Sact(\rho_N^{h,\xi}) \big)\:,
\eeq
showing that the length of the discrete curve is bounded by the total change of the action.
This estimate suggests that it is useful to use the action itself for the parametrization
of the curve. As we shall see, it is of advantage to do so already for the discrete curve, before
taking the limit~$h \searrow 0$ (as will be explained in Remark~\ref{remrepar} below).
To this end, given~$h,\xi>0$, we set
\beq \label{sjdef}
s_j = \Sact \big( \rho^{h, \xi}_j \big) \qquad \text{with~$j \in \N$} \:.
\eeq
Then the sequence~$(s_j)_{j \in \N}$ is monotone decreasing, $s_j \geq s_{j+1} \geq \cdots$.
Moreover, the estimate~\eqref{totvarest} shows that the measures~$\rho_j^{h,\xi}$ converge
in the limit~$j \rightarrow \infty$,
\[ \rho^{h, \xi}_j \rightarrow \rho^{h, \xi}_\infty \:, \]
and that the action is continuous, i.e.\
\[ s_j \searrow \Sact \big( \rho^{h, \xi}_\infty \big) \:. \]

We now define a continuous curve by interpolation,
\beq \label{interpolate}
\tilde{\rho}^{h, \xi}(s) := \frac{s_j-s}{s_j-s_{j+1}}\: \rho^{h, \xi}_j + \frac{s-s_{j+1}}{s_j-s_{j+1}}\: \rho^{h, \xi}_{j+1}  \qquad \text{if~$s \in [s_{j+1}, s_j]$}\:.
\eeq
This formula can be used even if~$s_j=s_{j+1}$, in which case
\[ \tilde{\rho}^{h, \xi}(s) = \rho^{h, \xi}_j = \rho^{h, \xi}_{j+1} \:. \]
In this way, we obtain a continuous curve of measures
\[ \tilde{\rho}^{h, \xi} \::\: \big[ \Sact \big( \rho^{h, \xi}_\infty \big), \Sact \big( \rho_0 \big) \big]
\rightarrow \meas_{1}(\F) \:. \]

\begin{lemma} \label{lemmarep} Assume that the Lagrangian is H\"older continuous,
$\L \in \hold^{0,\alpha}(\F \times \F, \R^+_0)$. Then there is a constant~$C>0$ (which depends only on~$\F$
and~$\L$) such that for all~$s,s' \in \big[ \Sact \big( \rho^{h, \xi}_\infty \big), \Sact \big( \rho_0 \big) \big]$
and~$h>0$,
\[ W_p \big( \tilde{\rho}^{h, \xi}(s) \big), \tilde{\rho}^{h, \xi}(s') \big)
\leq \frac{1}{\xi} \Big( |s-s'| + C \,h^\frac{\alpha}{2} \Big) \:. \]
\end{lemma}
\begin{proof} Given~$s$ and~$s'$ we choose~$j$ and~$k$ with
\[ s \in [s_{j+1}, s_j] \qquad \text{and} \qquad s' \in [s_{k+1}, s_k] \:. \]
Applying the triangle inequality as well as~\eqref{i1} yields
\begin{align*}
&W_p \big( \tilde{\rho}^{h, \xi}(s) \big), \tilde{\rho}^{h, \xi}(s') \big)
\leq W_p \big( \tilde{\rho}^{h, \xi}(s) \big), \tilde{\rho}^{h, \xi}(s_j) \big)
+ \frac{1}{\xi}\: \big| s_j - s_k \big| + W_p \big( \tilde{\rho}^{h, \xi}(s_k), \tilde{\rho}^{h, \xi}(s') \big) \\
&\leq W_p \big( \tilde{\rho}^{h, \xi}(s) \big), \tilde{\rho}^{h, \xi}(s_j) \big)
+ \frac{1}{\xi}\: \big| s_j - s \big| \\
&\quad\: + \frac{1}{\xi}\: \big| s - s' \big| + \frac{1}{\xi}\: \big| s' - s_k \big|
+ W_p \big( \tilde{\rho}^{h, \xi}(s_k), \tilde{\rho}^{h, \xi}(s') \big) \:.
\end{align*}

It remains to estimate the first two summands (the last summands can be treated similarly).
In order to estimate the first summand, we first apply Lemma~\ref{lem:wasserbound},
\[ W_p \big( \tilde{\rho}^{h, \xi}(s) \big), \tilde{\rho}^{h, \xi}(s_j) \big)
 \leq W_p(\rho^{h, \xi}_j, \rho^{h, \xi}_{j+1})
\leq d \big( \rho_j^{h,\xi},\rho_{j+1}^{h,\xi} \big) \overset{\eqref{i2}}{\leq} \sqrt{ 2 h\, \Sact(\rho_0) }\:. \]
The second summand can be estimated using~\eqref{lipes} (in which case we choose~$\alpha=1$)
or~\eqref{hoeles} by
\[ \frac{1}{\xi}\: \big| s_j - s \big| = 
\frac{1}{\xi}\: \Big| \Sact \big( \tilde{\rho}^{h, \xi}(s_j) \big) - \Sact \big( \tilde{\rho}^{h, \xi}(s) \big) \Big| 
\leq \frac{C}{\xi}\: d\big( \tilde{\rho}^{h, \xi}(s_j), \tilde{\rho}^{h, \xi}(s) \big)^\alpha \:. \]
Again Applying Lemma~\ref{lem:wasserbound} and~\eqref{hoeles} gives
\[ \frac{1}{\xi}\: \big| s_j - s \big| \leq \frac{C}{\xi}\: d \big( \rho_j^{h,\xi},\rho_{j+1}^{h,\xi} \big)^\alpha
\leq \frac{C}{\xi}\: \big( 2 h\, \Sact(\rho_0) \big)^\frac{\alpha}{2}\:. \]
This concludes the proof.
\end{proof}

After these preparations, we can take the limit~$h \searrow 0$ to obtain the following result.
\begin{proposition}\label{prop:LipRepara}
By iteratively choosing subsequences and taking the limit of the diagonal sequence,
one obtains a curve of measures denoted by
\beq \label{tilrhodef}
\tilde{\varrho}^{\xi} \::\: \big[ \Sact^\xi_{\min}, \Sact \big( \rho_0 \big) \big] \rightarrow \meas_{1}(\F) \:,
\eeq
where
\[ \Sact^\xi_{\min} := \liminf_{h \searrow 0} \Sact \big(
\tilde{\rho}^{h, \xi}_\infty \big) \:. \]
The curve~$\tilde{\varrho}^\xi(s)$ is Lipschitz continuous in the sense that
\begin{align}\label{eq:LipRePara}
 d\big( \tilde{\varrho}^\xi(s_2),\tilde{\varrho}^\xi(s_1) \big) \leq \frac{1}{\xi} \: \big( s_2 - s_1 \big) \qquad
\text{for all~$\Sact^\xi_{\min} \leq s_1 < s_2 \leq \Sact(\varrho_0)$}.
\end{align}

Moreover, there is a sequence~$h_\ell$ with~$h_\ell \searrow 0$ such that the
end points of the corresponding piecewise linear curves converge, i.e.\
\beq \label{basicconv}
\tilde{\rho}^{h_\ell, \xi} \Big( \Sact \big( \rho^{h_\ell, \xi}_\infty \big) \Big) 
\overset{\ell \rightarrow \infty}{\longrightarrow} \tilde{\varrho}^\xi \big( \Sact^\xi_{\min} \big) \:.
\eeq
\end{proposition}
\begin{proof}
We let~$(h_n)_{n \in \N}$ be a real sequence which is monotone decreasing and tends to zero,
\[ h_n \searrow 0 \:. \]
Moreover, we let~$(s_\ell)_{\ell \in \N}$ with
\[ s_\ell \in \big( \Sact^\xi_{\min}, \Sact(\rho_0) \big] \]
be a sequence which is dense in the last interval. Then for every~$\ell \in \N$, there is an infinite
number of~$h_n$ with the property that the piecewise linear curve is defined at~$s_\ell$, i.e.\
\[ s_\ell > \Sact \big( \tilde{\rho}^{h_n, \xi}_\infty \big) \:. \]
Using compactness of measures, there is a weak*-convergent subsequence with
\[ \tilde{\rho}^{h_{n_k}, \xi}(s_\ell) \overset{k \rightarrow \infty}{\longrightarrow} \tilde{\varrho}^{\xi}(s_\ell) \:. \]
We now proceed inductively in the parameter~$\ell = 1,2, \ldots$ and choose inductive subsequences.
For the resulting diagonal sequence, which for simplicity we denote again by~$h_{n_k}$,
the measures converge to a limit curve of measures, i.e.\
\[ \tilde{\rho}^{h_{n_k}, \xi}(s_\ell) \overset{k \rightarrow \infty}{\longrightarrow} \tilde{\varrho}^{\xi}(s_\ell) \qquad \text{for all~$\ell \in \N$}\:. \]
Considering the interpolation~\eqref{eq:curvedefineA}, applying the estimate~\eqref{i1} and passing to the limit, we find that 
the family of limit measures is again Lipschitz continuous in the sense that
\[ d \big( \tilde{\varrho}^{\xi}(s_\ell), \tilde{\varrho}^{\xi}(s_{\ell'})\big) \leq \frac{1}{\xi} \: \big| s_\ell - s_{\ell'} \big| \:. \]
Therefore, it extends by continuity to the curve~$\tilde{\varrho}^\xi$ in~\eqref{tilrhodef}
being Lipschitz continuous~\eqref{eq:LipRePara}.

In order to prove~\eqref{basicconv}, we estimate the Wasserstein distance
(which, as specified in~\eqref{eq:soundofwater1}, metrizes the weak*-topology).
We first note that, for any~$\ell \in \N$ and~$h>0$,
\begin{align}
&W_p\Big( \tilde{\rho}^{h, \xi} \big( \Sact \big( \rho^{h, \xi}_\infty \big) \big), \tilde{\varrho}^\xi \big( \Sact^\xi_{\min} \big) \Big) \notag \\
&\leq
W_p\Big( \tilde{\rho}^{h, \xi} \big( \Sact \big( \rho^{h, \xi}_\infty \big) \big), 
\tilde{\rho}^{h, \xi} (s_\ell) \Big) + 
W_p\Big( \tilde{\rho}^{h, \xi} (s_\ell), \tilde{\varrho}^\xi(s_\ell) \Big) + 
W_p\Big( \tilde{\varrho}^\xi(s_\ell), \tilde{\varrho}^\xi \big( \Sact^\xi_{\min} \big) \Big) \notag \\
&\leq \frac{1}{\xi}\: \Big( \Sact \big( \rho^{h, \xi}_\infty \big) - s_\ell + C\, h^\frac{\alpha}{2} \Big) 
+ W_p\Big( \tilde{\rho}^{h, \xi} (s_\ell), \tilde{\varrho}^\xi(s_\ell) \Big) + 
\frac{1}{\xi}\: \Big( \Sact^\xi_{\min} - s_\ell \Big) \:, \label{Ch}
\end{align}
where in the last step we applied Lemma~\ref{lemmarep}.
Choosing~$h=h_{n_k}$ as our diagonal sequence and passing to the limit, we obtain
\[ \liminf_{k \rightarrow \infty} 
W_p\Big( \tilde{\rho}^{h_{n_k}, \xi} \big( \Sact \big( \rho^{h_{n_k}, \xi}_\infty \big) \big), \tilde{\varrho}^\xi \big( \Sact^\xi_{\min} \big) \Big) \leq \frac{2}{\xi}\: \Big( \Sact^\xi_{\min} - s_\ell \Big) \:. \]
Taking the limit~$s_\ell \searrow \Sact^\xi_{\min}$ shows that~\eqref{basicconv} holds
(again for a suitable subsequence).
\end{proof}

\subsection{Limiting measures and Euler-Lagrange equations} \label{seclimit}
Based on the construction of curves of measures in the previous subsection, we now turn to their convergence properties. In particular, we are interested in whether the underlying curves converge and, if so, whether the limit measure satisfies the corresponding Euler-Lagrange equations at least approximately. 

In the case without $\xi$-penalization, we have the following result.
\begin{theorem} \label{thmlimit1} 
Consider the minimizing movement flow corresponding to the
action with penalization~\eqref{eq:varprinB0} and~\eqref{Sactdef},
where the Lagrangian~$\L$ has the properties~{\rm{(A1)}} and~{\rm{(A2)}}
stated in the preliminaries on page~\pageref{enumAB}.
In the case~$\xi=0$, assume that the curve~$\varrho^0(t)$ with initial measure~$\varrho^{0}(0)=\rho_{0}\in\mathfrak{M}_{1}(\F)$ converges in the weak*-sense. We set
\[ \varrho_\infty := \mathrm{w}^{*}\text{-}\lim_{t \rightarrow \infty} \varrho^0(t) \:. \]
Moreover, assume that for a sequence~$h_k$ with~$h_k \searrow 0$ the discrete sequences
converge,
\[ \rho_n^{h_k} \overset{n \rightarrow \infty}{\longrightarrow} \rho_\infty^{h_k} \:, \]
and that the limit measures converge to the limit point of the curve,
\[ \rho_\infty^{h_k} \overset{k \rightarrow \infty}{\longrightarrow} \varrho_\infty \:. \]
Then the measure~$\varrho_\infty$ satisfies the EL equations~\eqref{EL}.
\end{theorem}
Clearly, the assumptions on the existence of limits of measures in this theorem are quite strong
and restrictive. However, it seems impossible to relax these assumptions because, as explained in detail in
the example in Section~\ref{secex}, such a limit point will in general not exist.

In the case~$\xi>0$, the situation is much better, because the results of the preceding subsection imply that the underlying curves of measures have finite length. This, in turn, can be used to establish the following stronger result on the Euler-Lagrange equations being \emph{approximately} satisfied in the limit: 
\begin{theorem} {\em{(Convergence and approximative EL-equations)}} \label{thmlimit2} \\
Consider the minimizing movement flow corresponding to the
action with penalization~\eqref{eq:varprinB0} and~\eqref{Sactdef},
where the Lagrangian~$\L$ has the properties~{\rm{(A1)}} and~{\rm{(A2)}}
stated in the preliminaries on page~\pageref{enumAB}.
In the case~$\xi>0$,
the curve~$\varrho^\xi(s)$ converges as~$s \searrow \Sact^\xi_{\min}$.
In the case of penalization by the Wasserstein distance~$W_p$
(i.e., Case~2\ in~\eqref{dspecify}), the limiting measure
\[ \varrho^\xi_\infty := \lim_{s \searrow \Sact^\xi_{\min}} \varrho^\xi(s) \]
satisfies the EL equations approximately, in the sense that the function~$\ell_\xi$ defined by
\[ \ell_\xi(x) := \int_\F \L(x,y)\: \dif \varrho^\xi_\infty(y) + \frac{\xi}{2} \: \wasser(\delta_{z},\mu) \]
is minimal on~$N := \supp \varrho^\xi_\infty$,
\[ 
\ell_\xi|_N \equiv \inf_\F \ell_\xi \:. \]
\end{theorem}

The remainder of this section is devoted to the proofs of these theorems.
We alleviate notation by setting
\begin{align*}
&\alpha_{1}:= \Big( \Big\lfloor \frac{t}{h}+1 \Big\rfloor - \frac{t}{h} \Big),\;\;\;\alpha_{2}=\frac{t}{h}-\Big\lfloor
\frac{t}{h} \Big\rfloor
\qquad \text{and} \qquad
\rho_{(1)}:=\rho_{\lfloor \frac{t}{h}\rfloor},\;\;\;\rho_{(2)}:=\rho_{\lfloor\frac{t}{h}+1\rfloor} \:,
\end{align*}
so that the interpolated measure defined in~\eqref{eq:curvedefineA} can be written as
\begin{align*}
\rho^{h, \xi}(t):=\alpha_{1}(t)\,\rho_{(1)}(t)+\alpha_{2}(t)\,\rho_{(2)}(t) \:. 
\end{align*}
\begin{lemma} \label{lemmaELlim}
Let~$h>0$, $\xi\geq0$ and denote 
by~$\rho_\infty^{h,\xi}\in\mathfrak{M}_{1}(\F)$ a weak*-accumulation point of~$(\rho_j^{h,\xi})$ as~$j \rightarrow \infty$. Then, for all~$z\in\F$, we have 
\begin{align}\label{eq:ApproxElMain}
\iint_{\F\times\F}\L(x,y)\dif\rho_\infty^{h,\xi}(x)\dif\rho_\infty^{h,\xi}(y) & \leq \int_{\F}\L(x,z)\dif\rho_\infty^{h,\xi}(x)  + \frac{\xi}{2}\:\wasser(\delta_{z},\rho_\infty^{h,\xi}) \:.
\end{align} 
\end{lemma}
\begin{proof}
Given~$0<\tau<1$ and~$z\in\F$ we define
\begin{align*}
\mu_{\tau}^{j,h,\xi}:=(1-\tau)\rho_j^{h,\xi}+\tau\delta_{z}\in\meas_{1}(\F) \:.
\end{align*}
Using that~$\rho_j^{h,\xi}$ is a minimizer of the penalized action, it follows that
\begin{align*} 
0 &\leq \frac{1}{\tau}\Big(\mathcal{S}(\mu_{\tau}^{j,h,\xi})-\mathcal{S}(\rho_j^{h,\xi}) \Big) + \frac{1}{2\tau h}\Big(\wasser(\mu_{\tau}^{j,h,\xi},\rho_{j-1}^{h,\xi})^{2}-\wasser(\rho_j^{h,\xi},\rho_{j-1}^{h,\xi})^{2} \Big) \\
& \quad\:+ \frac{\xi}{\tau} \Big(\wasser(\mu_{\tau}^{j,h,\xi},\rho_{j-1}^{h,\xi})-\wasser(\rho_j^{h,\xi},\rho_{j-1}^{h,\xi}) \Big) \\ 
& \leq \frac{1}{\tau}\Big(\mathcal{S}(\mu_{\tau}^{j,h,\xi})-\mathcal{S}(\rho_j^{h,\xi}) \Big) + \frac{1}{2\tau h}\Big(\wasser(\mu_{\tau}^{j,h,\xi},\rho_{j-1}^{h,\xi})^{2}-\wasser(\rho_j^{h,\xi},\rho_{j-1}^{h,\xi})^{2} \Big) \\
& \quad\:+ \frac{\xi}{\tau} \wasser(\mu_{\tau}^{j,h,\xi},\rho_j^{h,\xi}) =: \mathrm{IV} + \mathrm{V} + \mathrm{VI}
\end{align*}
by use of the triangle inequality. By assumption, we have 
\begin{align*}
\mu_{\tau}^{j,h,\xi}\stackrel{*}{\rightharpoonup} \mu_{\tau}^{\infty,h,\xi}:=(1-\tau)\rho_\infty^{h,\xi}+\tau\delta_{z} \:, 
\end{align*}
whereby Lemma~\ref{lem:weakconv} yields that 
\begin{align}\label{eq:Vbound0}
\mathrm{IV} \to \frac{1}{\tau}\Big(\mathcal{S}( \mu_{\tau}^{\infty,h,\xi})-\mathcal{S}(\rho_\infty^{h,\xi}) \Big) \qquad \text{as~$j\to\infty$}\:. 
\end{align}
For term~$\mathrm{V}$, we use Lemma~\ref{lem:wasserbound} to estimate and expand terms as follows,
\begin{align*} 
\mathrm{V} & \leq \frac{1}{2\tau h}\Big( \big( (1-\tau)\wasser(\rho_j^{h,\xi},\rho_{j-1}^{h,\xi})+\tau\wasser(\delta_{z},\rho_{j-1}^{h,\xi}) \big)^{2}-\wasser(\rho_j^{h,\xi},\rho_{j-1}^{h,\xi})^{2} \Big) \\ 
& = \frac{1}{2h}\Big(-2\wasser(\rho_j^{h,\xi},\rho_{j-1}^{h,\xi})^{2} + \tau\wasser(\rho_j^{h,\xi},\rho_{j-1}^{h,\xi})^{2} \Big.\\ 
& \Big. \;\;\;\;\;\;\;\; + 2(1-\tau)\:\wasser(\rho_j^{h,\xi},\rho_{j-1}^{h,\xi})\:\wasser(\delta_{z},\rho_{j-1}^{h,\xi}) + \tau\wasser(\delta_{z},\rho_{j-1}^{h,\xi})^{2}\Big) \:. 
\end{align*}
Since~$\xi \geq 0$ and~$h>0$ are fixed, we have that~$\wasser(\rho_j^{h,\xi},\rho_{j-1}^{h,\xi})\to 0$ as~$j\to\infty$. Moreover, $\sup_{j\in\mathbb{N}}\wasser(\delta_{z},\rho_{j-1}^{h,\xi})<\infty$, and therefore 
\[ 
\limsup_{j\to\infty}\mathrm{V} \leq \frac{\tau}{2h}\wasser(\delta_{z},\rho_\infty^{h,\xi}) \:. \]
Lastly, employing Lemma~\ref{lem:wasserbound}, we arrive at the following estimate for~$\mathrm{VI}$: 
\begin{align}\label{eq:Vbound2}
\mathrm{VI} & \leq \xi\wasser(\delta_{z},\rho_j^{h,\xi}) \to  \xi\wasser(\delta_{z},\rho_\infty^{h,\xi})
\end{align} 
as~$j\to\infty$. Combining~\eqref{eq:Vbound0}--\eqref{eq:Vbound2}, we obtain 
\begin{align}\label{eq:Vbound3}
0 \leq  \frac{1}{\tau}\Big(\mathcal{S}( \mu_{\tau}^{\infty,h,\xi})-\mathcal{S}(\rho_\infty^{h,\xi}) \Big)  + \frac{\tau}{2h}\wasser(\delta_{z},\rho_\infty^{h,\xi}) + \xi\wasser(\delta_{z},\rho_\infty^{h,\xi}).
\end{align}
At this stage, we aim to send~$\tau\searrow 0$. Working from~\eqref{eq:Vbound3}, we expand using the symmetry of~$\L$,
\begin{align*}
0 & \leq \frac{(1-\tau)^{2}}{\tau}\iint_{\F\times\F}\L(x,y)\dif\rho_\infty^{h,\xi}(x)\dif\rho_\infty^{h,\xi}(y)  - \frac{1}{\tau}\iint_{\F\times\F}\L(x,y)\dif\rho_\infty^{h,\xi}(x)\dif\rho_\infty^{h,\xi}(y) \\ 
& \quad\: + 2(1-\tau)\int_{\F}\L(x,z)\dif\rho_\infty^{h,\xi}(x) + \tau \L(z,z) + \frac{\tau}{2h}\wasser(\delta_{z},\rho_\infty^{h,\xi}) + \xi\wasser(\delta_{z},\rho_\infty^{h,\xi})\\ 
& \xrightarrow{\tau\searrow 0}
 -2\iint_{\F\times\F}\L(x,y)\dif\rho_\infty^{h,\xi}(x)\dif\rho_\infty^{h,\xi}(y) + 2\int_{\F}\L(x,z)\dif\rho_\infty^{h,\xi}(x) + \xi\wasser(\delta_{z},\rho_\infty^{h,\xi}) \:.
\end{align*}
Hence, we arrive at 
\begin{align*}
\iint_{\F\times\F}\L(x,y)\dif\rho_\infty^{h,\xi}(x)\dif\rho_\infty^{h,\xi}(y) & \leq \int_{\F}\L(x,z)\dif\rho_\infty^{h,\xi}(x)  + \frac{\xi}{2}\wasser(\delta_{z},\rho_\infty^{h,\xi}) \:.
\end{align*} 
This is~\eqref{eq:ApproxElMain}, and the proof is complete. 
\end{proof}

\begin{lemma}\label{lem:verticaldown}
Let~$\xi \geq 0$, and denote by~$\meas^\infty$ 
the set of all weak*-accumulation points of~$(\rho_\infty^{h,\xi})$ as~$h\searrow 0$. Whenever~$\mu\in \meas^\infty$ and~$z\in\F$ are such that 
\beq \label{eq:minimiserinz}
\inf_{y\in\F}\int_{\F}\L(x,y)\dif\mu(x)=\int_{\F}\L(x,z)\dif\mu(x) \:, 
\eeq
we have 
\beq \label{eq:ApproxElMain1}
\iint_{\F\times\F}\L(x,y)\dif\mu(x)\dif\mu(y) \leq \inf_{y\in\F}\int_{\F}\L(x,y)\dif\mu(x) + \frac{\xi}{2}\:\wasser(\delta_{z},\mu) \:. 
\eeq
\end{lemma} 
\begin{proof} 
Since~$\F$ is compact and the right-hand side of~\eqref{eq:minimiserinz} is a continuous function in the second variable, we find~$z\in\F$ such that~\eqref{eq:minimiserinz} is satisfied.  We let~$(h_{k})\subset\R_{>0}$ be a sequence with~$h_{k}\searrow 0$ and~$\rho_\infty^{h_{k},\xi}\stackrel{*}{\rightharpoonup}\mu$ as~$k\to\infty$. By the continuity result from Lemma~\ref{lem:weakconv}, it is then clear that the left-hand side of~\eqref{eq:ApproxElMain} converges to the left-hand side of~\eqref{eq:ApproxElMain1}. On the other hand, since~$\wasser$ metrizes weak*-convergence, we also have~$\wasser(\delta_{z},\rho_\infty^{h_{k},\xi})\stackrel{*}{\rightharpoonup}\wasser(\delta_{z},\mu)$ as~$k\to\infty$. Again using continuity under
weak* convergence, we obtain
\begin{align*}
\int_{\F}\L(x,z)\dif\mu(x)=\lim_{k\to\infty}\int_{\F}\L(x,z)\dif\rho_\infty^{h_{k},\xi}(x),\qquad k\to\infty,  
\end{align*}
and then~\eqref{eq:ApproxElMain1} follows at once. 
\end{proof} 
Based on these preparations, we can now prove the main results of this section.
\begin{proof}[Proof of Theorems~\ref{thmlimit1} and~\ref{thmlimit2}]
According to Lemma~\ref{lem:verticaldown}, it suffices to show that there is a sequence~$(h_k)_{k \in \N}$
with~$h_k \searrow 0$ such that the corresponding discrete limit measures~$(\rho_\infty^{h_k,\xi})$
converge to the limit measure~$\varrho_\infty$ respectively~$\varrho^\xi_\infty$.
In the case~$\xi=0$, this is a consequence of the assumptions in Theorem~\ref{thmlimit1}.
In the case~$\xi>0$, on the other hand, this was proved in~\eqref{basicconv}.
\end{proof}

\begin{remark} (Why the reparametrization) \label{remrepar} {\em{ At the beginning of Section~\ref{seclip}, we reparametrized
the discrete curve by the action (see~\eqref{sjdef}). After interpolating~\eqref{interpolate}
and taking the limit~$h \searrow 0$, we obtained a continuous curve~$\varrho^\xi(s)$,
where the parameter~$s$ coincides with the action along the curve.

The purpose of the reparametrization by the action is to avoid energy plateaus, as we now
explain. Suppose we had taken the limit~$h \searrow 0$ without reparametrizing.
Then it is a possible scenario that the corresponding interpolated curve~$\rho^{h, \xi}(t)$
defined by~\eqref{eq:curvedefineA} stays almost constant for a certain range of the parameter~$t$
before leaving the energy plateau and approaching the minimizer at~$\Sact^{h, \xi}(\infty)$
(see Figure~\ref{fig:reparamterise}).
\begin{figure} 
	\begin{center}
		\begin{tikzpicture} 
			\node [left] at (-0.15,2) {$\mathcal{S}(\rho_{0})$}; 
			\node [left] at (-0.15,0.22) {$\mathcal{S}^{h,\xi}(\infty)$}; 
			\node [left] at (5.5,0) {$t$}; 
			\draw[->, thick] (0,-0.25) -- (0,2.5); 
			\draw[->, thick] (-0.25,0) -- (5,0); 
			\draw[-, thick] (0,2.025) to (1,2) [out = 0, in =180] to (2,1) -- (3,1) [out = 0, in = 180] to (5,0.25);
			\draw[green!40!black, ultra thick] (2,1) -- (3,1);
			\draw[-,green!40!black] (2,0.9) -- (2,1.1); 
			\draw[-,green!40!black] (3,0.9) -- (3,1.1); 
			\draw[dashed] (0,0.22) -- (5,0.22);
		\end{tikzpicture} 
	\end{center}
	\caption{Possible energy profile in the un-reparametrized situation. The reparametrization lets the flow clear such plateaus where the energy is not strictly decreased.} \label{fig:reparamterise}
\end{figure}
Since we have no a-priori control on the size of this parameter range in~$t$, we cannot exclude
the situation that the time~$t=t(h)$ when the curve leaves the plateau tends to infinity as~$h \searrow 0$.
In this case, the limiting curve as~$h \searrow 0$ would remain on the plateau for all~$t$,
implying that the end points~$\rho^{h, \xi}(\infty)$ would not converge as~$h \searrow 0$.
As a consequence, we could not be clear how to prove that the limit measure~$\rho^\xi(\infty)$
satisfies the approximative EL equations.

After the reparametrization by the action, however, the corresponding interpolated curves~\eqref{interpolate}
leave the energy plateau at a parameter~$s$ uniformly in~$h$, giving the desired convergence of the end points~\eqref{basicconv}.
This is crucial for proving that the limit measure satisfies the approximative EL equations
(Theorem~\ref{thmlimit2}).
}} \QEDrem
\end{remark}

\section{Further examples} \label{secfurtherex}
We now illustrate the previous abstract results in a few examples.
We choose~$\F= \overline{B_1(0)} \subset \R^2$ as a closed unit ball
in two dimensions. Moreover, we choose~$x_{0} \in \F$ and let~$\rho_0:=\delta_{x_{0}}$ be the Dirac measure at~$x_{0}$. 
Given a bounded continuous function~$V \in \hold^0(\F, \R) \cap L^\infty(\F, \R)$, we define the Lagrangian by
\begin{align*}
\L(x,y):=\frac{1}{2}\:\big( V(x)+V(y) \big)+c\:|x-y|^{2},\qquad x,y\in\F \:.
\end{align*}
The corresponding penalized action reads
\beq \label{penalty}
\begin{split}
\mathcal{S}^{h,\xi}(\mu) & =\int_{\F}V(x)\dif\mu(x) + c
\int_{\F} \dif\mu(x) \int_{\F} \dif\mu(y) \: |x-y|^{2} \\
&\quad\; + \frac{1}{2h}\:d(\mu,\rho_0)^{2}+\xi \:d(\mu,\rho_0) \:,
\end{split}
\eeq
where~$d$ is again either the Fr{\'e}chet or the Wasserstein metric~\eqref{dspecify}.

We begin with the case of the Wasserstein distance.
\begin{lemma} Assume that~$d=W_{p}$ for some~$2 \leq p<\infty$. Then, for any~$c>0$, every minimizer
of the penalized action~\eqref{penalty} has the form~$\rho=\delta_{x_{1}}$ for some~$x_{1}\in\F$.
\end{lemma}
\begin{proof}
We observe that, for a Dirac measure centered at some~$x\in\F$, the
penalized action simplifies to
\begin{align}\label{eq:valueexplicit}
\mathcal{S}^{h,\xi}(\delta_x)=V(x) +\frac{1}{2h}|x-x_0|^{2} +\xi \:|x-x_0| \:. 
\end{align}
Since~$V$ is bounded and continuous, this function is minimal at some~$x_{1}\in\F$.
Next, let~$\rho\in\meas_{1}(\F)$ be an arbitrary measure. Using that
\[ d(\rho,\rho_0) = \bigg( \int_\F |x-x_0|^p \: \dif \rho(x) \bigg)^\frac{1}{p}\:, \]
we obtain
\begin{align}
\mathcal{S}^{h,\xi}(\rho) &= \int_\F \Big( V(x) +\frac{1}{2h}|x-x_0|^{2} +\xi\:|x-x_0| \Big)\: \dif \rho(x) \label{t1} \\
&\quad\:+ c \int_{\F} \dif\rho(x) \int_{\F} \dif\rho(y) \: |x-y|^{2} \label{t2} \\
&\quad\:+ \frac{1}{2h} \bigg(\int_\F |x-x_0|^p\: \dif \rho(x) \bigg)^\frac{2}{p} - \frac{1}{2h} \int_\F |x-x_0|^{2}\:
\dif \rho(x) \label{t3} \\
&\quad\:+ \xi \bigg(\int_\F |x-x_0|^p\: \dif \rho(x) \bigg)^\frac{1}{p} - \xi \int_\F |x-x_0|\: \dif \rho(x) \label{t4} \:.
\end{align}
Now~\eqref{t1} is bounded from below by~$\mathcal{S}^{h,\xi}(\delta_{x_1})$
(recall that~$x_1$ was defined as the minimizer of the integrand of~\eqref{t1}).
Moreover, \eqref{t2} is obviously non-negative, and it is zero if and only if~$\rho$ is a Dirac measure.
Finally, the summands in~\eqref{t3} and~\eqref{t4} are non-negative in view of H\"older's inequality for normalized measures
(here we make essential use of the fact that~$p \geq 2$).
We conclude that every minimizing measure is a Dirac measure.
\end{proof}
In view of this lemma, the flow constructed in Section~\ref{secflow} reduces to the flow obtained
by minimizing movements from the action in the plane~\eqref{eq:valueexplicit}.
If~$V$ is smooth and~$\xi=0$, we obtain the usual gradient flow for a curve~$\gamma$ in~$\F$
\[ \dot{\gamma}(t) = -\nabla V\big( \gamma(t) \big) \:. \]

The above example generalizes immediately to higher dimension. In this way,
any gradient flow in finite dimension can be recovered as a minimizing movement flow
of a specific class of causal variational principles.

The above example changes considerably in the 
case~$d=d_{\meas(\F)}$ where we penalize with the Fr{\'e}chet metric.
In this case, for a Dirac measure, the action becomes
\[ \mathcal{S}^{h,\xi}(\delta_x)=V(x) + \left\{ \begin{array}{cl} 0 & \text{if~$x=x_0$} \\
\displaystyle \frac{1}{2h} +\xi & \text{if~$x \neq x_0$}\:. \end{array} \right. \]
Minimizing this action for sufficiently small~$h$, we get the unique minimizer~$\mu=\rho_0$.
Therefore, considering minimizing movements in the class of Dirac measures gives
the constant flow~$\varrho(t)=\rho_0$. This flow converges trivially in the limit~$t \rightarrow \infty$,
but the limit measure does not need to satisfy any EL equations or approximative EL equations.

Nevertheless, minimizing movements become non-trivial if one varies in the class~$\rho\in\meas_{1}(\F)$ of
arbitrary measures. To see this, we let~$x_1$ be a minimum of the potential~$V$.
We consider the family of measures~$(\rho_\tau)_{\tau \in [0,1]}$ with
\[ \rho_\tau = \tau\: \delta_{x_1} + (1-\tau)\: \delta_{x_0} \:. \]
Then
\[ \mathcal{S}^{h,\xi}(\rho_\tau)= V(x_0) + \tau \:\big(V(x_1) - V(x_0) \big) +
2 c\, \tau (1-\tau)\: |x_1-x_0|^2 + \frac{1}{2h}\: \tau^2 + \xi \,\tau \:. \]
Note that the linear term~$\tau (V(x_1) - V(x_0))$ is negative.
This implies that the minimizer within our family is attained for~$\tau>0$, provided that~$c$
and~$\xi$ are sufficiently small.
The flow constructed in Section~\ref{secflow} is non-local in the sense that the support of~$\varrho(t)$
typically changes discontinuously. This can be understood immediately from the fact that the 
total variation norm does not involve the metric on~$\F$ and therefore cannot ``see'' if the points on~$\F$
are near or far apart. Nevertheless, as is made precise in Section~\ref{seclimit},
this nonlocal flow tends to a critical measure.

\section{Minimizing movements for causal fermion systems in finite dimensions}\label{seccfs}
The goal of this section is to extend the previous constructions to the causal action principle
for causal fermion systems on a finite-dimensional Hilbert space.
\subsection{Causal fermion systems and the reduced causal action principle}
We now recall the basic setup and introduce the main objects to be used later on.
\begin{definition} \label{defcfs} (causal fermion systems of fixed trace) {\em{ 
Given a finite-dimensional Hilbert space~$\H$ with scalar product~$\la .|. \ra_\H$
and a parameter~$n \in \N$ (the {\em{``spin dimension''}}), we 
let~$\F \subset \Lin(\H)$ be the set of all
symmetric linear operators~$x$ on~$\H$ with trace one,
\beq \label{fixedtrace}
\tr x = 1 \:,
\eeq
which (counting multiplicities) have at most~$n$ positive and at most~$n$ negative eigenvalues.
On~$\F$ we are given
a positive measure~$\rho$ (defined on a~$\sigma$-algebra of subsets of~$\F$).
We refer to~$(\H, \F, \rho)$ as a {\em{causal fermion system}}.
}}
\end{definition} \noindent

A causal fermion system describes a spacetime together
with all structures and objects therein.
In order to single out the physically admissible
causal fermion systems, one must formulate physical equations. To this end, we impose that
the measure~$\rho$ should be a minimizer of the causal action principle,
which we now introduce. For any~$x, y \in \F$, the product~$x y$ is an operator of rank at most~$2n$. 
However, in general it is no longer a symmetric operator because~$(xy)^* = yx$,
and this is different from~$xy$ unless~$x$ and~$y$ commute.
As a consequence, the eigenvalues of the operator~$xy$ are in general complex.
We denote these eigenvalues counting algebraic multiplicities
by~$\lambda^{xy}_1, \ldots, \lambda^{xy}_{2n} \in \C$
(more specifically,
denoting the rank of~$xy$ by~$k \leq 2n$, we choose~$\lambda^{xy}_1, \ldots, \lambda^{xy}_{k}$ as all
the non-zero eigenvalues and set~$\lambda^{xy}_{k+1}, \ldots, \lambda^{xy}_{2n}=0$).
Given a parameter~$\kappa>0$ (which will be kept fixed throughout),
we introduce the~$\kappa$-Lagrangian and the causal action by
\begin{align}
\text{\em{$\kappa$-Lagrangian:}} && \L(x,y) &= 
\frac{1}{4n} \sum_{i,j=1}^{2n} \Big( \big|\lambda^{xy}_i \big|
- \big|\lambda^{xy}_j \big| \Big)^2 + \kappa\: \bigg( \sum_{j=1}^{2n} \big|\lambda^{xy}_j \big| \bigg)^2 \label{Lagrange} \\
\text{\em{causal action:}} && \Sact(\rho) &= \iint_{\F \times \F} \L(x,y)\: \dif \rho(x)\, \dif \rho(y) \:. \label{Sdef}
\end{align}
The {\em{reduced causal action principle}} is to minimize~$\Sact$ by varying the measure~$\rho$
under the
\[ 
\text{\em{volume constraint:}} \qquad \rho(\F) = 1 \:, \]
within the class of all regular Borel measures (with respect to the topology on~$\F \subset \Lin(\H)$ induced by
the operator norm).

In order to put these definitions into context, we briefly explain how the
above variational principle is obtained from the general causal action principle as
introduced in~\cite[\S1.1.1]{cfs}. First of all, we here restrict attention to the {\em{finite-dimensional
case}}~$\dim \H< \infty$. In this case, the total volume~$\rho(\F)$ is finite.
Using the rescaling freedom~$\rho \rightarrow \sigma \rho$, it is no loss of generality to restrict
attention to normalized measures.
Next, using that minimizing measures are supported on operators of constant trace
(see~\cite[Proposition~1.4.1]{cfs}), we may fix the trace of the operators. Moreover, 
by rescaling the operators according to~$x \rightarrow \lambda x$ with~$\lambda \in \R$,
one can assume without loss of generality that this trace is equal to one~\eqref{fixedtrace}.
Finally, we here consider the {\em{reduced}} variational principle where
the so-called boundedness constraint of the causal action principle is built in
by a a Lagrange multiplier term, namely the last summand in~\eqref{Lagrange}.
This Lagrange multiplier term is needed for the existence theory, which we now recall.

\subsection{Moment measures and existence theory}
Endowed with the metric induced by the operator norm,
\[ d(x,y) := \|x-y\|_{\Lin(\H)} \:, \]
the set~$\F \subset \Lin(\H)$ is a locally compact metric space.
However, it is unbounded and therefore {\em{not compact}}.
For this reason, the causal action principle does not quite fit to the compact setting as introduced
in Section~\ref{seccvp}. Nevertheless, we can adapt the methods, as we now explain.
The main tool is to work with the so-called moment measures first introduced in~\cite{continuum}.
\begin{definition} \label{defmm}
Let~$\K$ be the compact metric space
\[ \K = \{ p \in \F \text{ with } \|p\|=1 \} \cup \{0\} \:. \]
For a given measure~$\rho$ on~$\F$, we define the measurable sets~$\Omega \subset \K$ by the
requirement that the sets~$\R^+ \Omega = \{ \lambda p \:|\: \lambda \in \R^+, p \in \Omega\}$
and~$\R^- \Omega$ should be~$\rho$-measurable in~$\F$. We introduce the measures~$\m^{(0)}$, 
$\m^{(1)}_\pm$ and~$\m^{(2)}$ by
\begin{align*}
\m^{(0)}(\Omega) &= \frac{1}{2}\: \rho \big(\R^+ \Omega \setminus \{0\} \big) 
+ \frac{1}{2}\: \rho \big( \R^- \Omega \setminus \{0\} \big)
+ \rho \big( \Omega \cap \{0\} \big) \\ 
\m^{(1)}_+(\Omega) &= \frac{1}{2} \int_{\R^+ \Omega} \|p\| \,\dif \rho(p) \\ 
\m^{(1)}_-(\Omega) &= \frac{1}{2} \int_{\R^- \Omega} \|p\| \,\dif \rho(p) \\ 
\m^{(2)}(\Omega) &= \frac{1}{2} \int_{\R^+ \Omega} \|p\|^2 \,\dif \rho(p) \:+\:
\frac{1}{2} \int_{\R^- \Omega} \|p\|^2 \,\dif \rho(p)
\:. 
\end{align*}
The measures~$\m^{(l)}$ and~$\m^{(l)}_\pm$ are referred to as the~$l^\text{th}$ {\bf{moment measures}}.
\end{definition}
The main point is that the causal action as well as the constraints can be expressed purely in terms
of the moment measures. Indeed, 
as shown in~\cite[Section~2.3]{continuum} (for more details
see also~\cite[Section~12.6]{intro}),
the volume constraint~$\rho(\F)=1$ and the trace constraints
can be expressed as
\beq \label{m0}
\m^{(0)}(\K) = 1 \qquad \text{and} \qquad \tr(p)\: \dif \m^{(1)}(p) = \dif \m^{(0)}(p) \:,
\eeq
whereas the action~\eqref{Sdef} can be written as
\beq
\Sact(\rho) = \iint_{\K \times \K} \L(p, q)\: \dif \m^{(2)}(p) \,\dif \m^{(2)}(q) \label{Sm2} \:.
\eeq
Here we make essential use of the fact that the trace is homogeneous of degree one
and that the~$\kappa$-Lagrangian in both arguments is homogeneous of degree two.

Working with these moment measures, one can prove existence of minimizers,
as is summarized in the following theorem.
\begin{theorem} \label{thmexist}
Let~$(\rho_\ell)_{\ell \in \N}$ be a minimizing sequence. Then 
there exists a subsequence~$(\rho_{\ell_k})_{k \in \N}$ which
converges in the weak*-topology to a minimizer~$\rho$.
\end{theorem}
\begin{proof}
The proof is a direct adaptation of methods introduced in~\cite[Section~2]{continuum}
(see also~\cite[Section~12.6]{intro}). We only give a sketch and refer for more details to the
just-mentioned works. We let~$\m^{(l)_\ell}$ and~$\m^{(l)}_{\pm, \ell}$ 
be the moment measures corresponding to the measures~$\rho_\ell$.
Clearly, the measures~$\m^{(0)}_\ell$ and~$\m^{(1)}$ satisfy the constraints~\eqref{m0}.
Moreover, a direct estimate using the Lagrange multiplier term in~\eqref{Lagrange} shows that
the first and second moment measure are uniformly bounded.
Therefore, the Banach-Alaoglu theorem provides us with a non-relabeled subsequence such that
\[ \m_{\ell_k}^{(0)} \rightarrow \m^{(0)} \:,\qquad \m_{\ell_k,\pm}^{(1)} \rightarrow \m^{(1)}_\pm \qquad \text{and} \qquad \m_{\ell_k}^{(2)} \rightarrow \m^{(2)} \:. \]
with convergence in
the~$\hold^0(\K)^*$-topology, where~$\m^{(0)}\in \meas_{1}(\K)$
is a normalized Borel measure and~$\m^{(1)}_\pm, \m^{(2)} \in \meas(\K)$ are Borel measures. As shown in~\cite[Lemma~2.12]{continuum} (for more details see also~\cite[Chapter~12]{intro}),
we know that there is a parameter~$\varepsilon$ (which depends only on the spin dimension~$n$ and the
dimension of the Hilbert space~$f$) such that for any measurable set~$\Omega \subset \K$
the following inequalities hold,
\begin{align}
\m^{(1)}_\pm(\Omega)^2 &\leq \m^{(0)}(\Omega)\:\m^{(2)}(\Omega) \label{m1es} \\
\m^{(2)}(\K) &\leq\; \frac{\sqrt{\Sact(\rho)}}{\sqrt{\kappa}\: \varepsilon}\:. \label{m2es}
\end{align}
These inequalities show that the measures~$\m^{(2)}$ and~$\m^{(1)}_\pm$
are bounded. Therefore, we can introduce the signed measure~$\m^{(1)}$
by~$\m^{(1)} := \m^{(1)}_+ - \m^{(1)}_-$.
The estimate~\eqref{m1es} implies that this signed measure
is absolutely continuous with respect
to~$\m^{(0)}$. Therefore, it has the Radon-Nikodym representation
\beq \label{RN}
\m^{(1)} = f\: \m^{(0)}\qquad \text{with~$f \in L^1(\K, \dif \m^{(0)})$}\:.
\eeq
Moreover, we conclude from~\eqref{m2es} that~$f$ lies even in~$L^2(\K, \dif \m^{(0)})$ and that
\[ |f|^2\, \m^{(0)}\leq \m^{(2)} \:. \]
Since the $\kappa$-Lagrangian is non-negative, the action becomes smaller if we replace the
measure~$\m^{(2)}$ by~$|f|^2\, \m^{(0)}$. Therefore, the 
measure~$\rho$ defined by
\beq \label{rholimit}
\rho := F_* \m^{(0)}\qquad \text{with} \qquad F : \K \rightarrow \F \:,\quad x \mapsto f(x)\, x
\eeq
is the desired minimizer.
\end{proof}

We point out that the compactness result used in this proof yields convergent sequences of measures
\beq \label{mweak}
\m_\ell^{(0)} \rightarrow \m^{(0)} \qquad \text{and} \qquad \m^{(1)}_\ell \rightarrow \m^{(1)} \:.
\eeq
The action is {\em{lower semicontinuous}} with respect to this convergence, i.e.\
\beq \label{Slower}
\Sact(\rho) \leq \liminf_{\ell \rightarrow \infty} \Sact(\rho_\ell)
\eeq
with~$\rho$ as defined by~\eqref{rholimit} via the Radon-Nikodym decomposition~\eqref{RN}.

\subsection{Minimizing movements for the causal action principle}
In view of the constructions of the previous section, it seems preferable to work with
the moment measures.
For notational simplicity, we denote the zeroth moment measure by~$\m$.
Then the proof of Theorem~\ref{thmexist} shows that, for constructing
minimizers, it is no loss of generality to consider measures of the form
\beq \label{rhopush}
\rho = F_* \m
\eeq
with
\[ F : \K \rightarrow \F \:,\qquad x \mapsto f(x)\, x \qquad \text{with} \qquad f \in L^2(\K, d\m; \R^+_0) \:. \]
According to~\eqref{m0}, the volume and trace constraints are implemented by demanding that
\[ \m(\K) = 1 \qquad \text{and} \qquad f(x)\, \tr(x) = 1 \quad \text{for almost all~$x \in \K$} \:. \]
Moreover, according to~\eqref{Sm2}, the causal action becomes
\[ \Sact(\m, f) = \iint_{\K \times \K} \L(p, q)\: |f(p)|^2\: |f(q)|^2\: \dif \m(p) \,\dif \m(q) \:. \]
Note that the measure~$\rho$ is now described by the pair
\beq \label{PKdef}
(\m, f) \;\in\; {\mathcal{P}}(\K) := \big\{ (\mu, g) \:\big|\: \mu \in \meas_1(\K), \; g \in L^2(\K, \R^+_0; \dif \mu) \big\} \:.
\eeq

Guided by the procedure for causal variational principles~\eqref{eq:varprinB0}, we now want to
penalize the action. However, the choice of the distance function is not obvious.
A natural idea is to take the distance function which reproduces the topology of the convergence
of measures in~\eqref{mweak}. Since we now restrict attention to measures of the form~\eqref{rhopush},
the resulting distance function could be written as
\[ d\big( (\m, f), (\m', f') \big) := d(\m, \m') + d\big( f \,\m, f' \,\m' \big) \:, \]
where on the right we consider again the Fr{\'e}chet or the Wasserstein metric~\eqref{dspecify},
but now on~$\meas(\K)$. But this choice has the disadvantage that the action is only
lower semicontinuous~\eqref{Slower} (which would not allow for passing to the limit
in the EL equations, as done for causal variational principles in Lemma~\ref{lemmaELlim}).
Therefore, it is preferable to choose a parameter
\[ 
q > 2 \]
and to introduce a distance function on~${\mathcal{P}}(\K)$ by
\beq \label{d2}
d\big( (\m, f), (\m', f') \big) := d(\m, \m') + d\big( |f|^q \,\m, |f'|^q \,\m' \big) \:.
\eeq

In analogy to~\eqref{eq:varprinB}, given parameters~$\xi \geq 0$, $h>0$ and a
pair~$(\m_0, f_0) \in {\mathcal{P}}(\K)$, we consider the {\em{causal action with penalization}}
\begin{align}\label{eq:varprinB'}
\Sact^{h,\xi}(\m, f):=\Sact(\m, f)+\frac{1}{2h}\: d\big( (\m, f), (\m_0, f_0) \big)^2+
\xi\: d\big( (\m, f), (\m_0, f_0) \big)
\end{align}

\begin{lemma}\label{lem:exist1cfs}
For any~$q>2$, $\xi\geq 0$, $h>0$ and~$(\m_0, f_0) \in {\mathcal{P}}(\K)$, there exists a minimizer~$(\m, f) \in {\mathcal{P}}(\K)$ of the causal action with penalization~\eqref{eq:varprinB'}. Moreover, the action is continuous in the sense that every
minimizing sequence has a subsequence~$(\m_\ell, f_\ell)$ such that
\beq \label{Scont}
\Sact^{h,\xi}(\m, f) = \lim_{\ell \rightarrow \infty} \Sact^{h,\xi} \big( \m_\ell, f_\ell \big) \:.
\eeq
\end{lemma}
\begin{proof} Since the $\kappa$-Lagrangian is non-negative, the penalized action is bounded below and
thus~$m:=\inf \mathcal{S}^{h,\xi}$ exists in~$[0,\infty)$. We choose a minimizing sequence~$(\m_{\ell}, f_\ell)$
for~$\mathcal{S}^{h,\xi}$, so that~$m=\lim_{\ell\to\infty}\mathcal{S}^{h,\xi}(\m_{\ell}, f_\ell)$.
Due to the penalization, the sequences of measures~$\m_\ell$ and~$|f_\ell^q|\, \m_\ell$ are bounded. 
Therefore, the Banach-Alaoglu theorem provides us with a non-relabeled subsequence such that
\[ \m_\ell \rightarrow \m\:, \qquad |f_\ell|^q\, \m_\ell \rightarrow \m^{(q)} \]
with a normalized Borel measure~$\m \in \meas_{1}(\K)$ and a Borel measure~$\m^{(q)} \in \meas(\K)$.
Now for any Borel subset~$\Omega \subset \K$, we can apply the H\"older inequality to obtain
\[ \m^{(2)}_\ell(\Omega) = \int_\Omega f_\ell^2 \: \dif \m_\ell
\leq \m_\ell(\Omega)^{\frac{q-2}{q}} \bigg( \int_\Omega f_\ell^q\: \dif \m_\ell \bigg)^\frac{2}{q} \:. \]
Passing to the limit, we obtain
\[ \m^{(2)}(\Omega) 
\leq \m(\Omega)^{\frac{q-2}{q}}\: \m^{(q)}(\Omega)^{\frac{2}{q}} \: \]
This shows that~$\m^{(2)}$ is absolutely continuous with respect to~$\m$.
Therefore, we can represent it as~$\m^{(2)} = h\, \m$ with~$h \in L^1(\K, \dif \m)$.
Repeating this procedure for~$\m^{(1)}$, we conclude that there is a function~$f \in L^2(\K, \dif \m)$ such that
\[ \m^{(1)} = f\, \mu \qquad \text{and} \qquad \m^{(2)} = f^2\, \m \:. \]
Therefore, defining the limit measure~$\rho$ again by~\eqref{rholimit},
all the moment measures~$\m_\ell$, $\m^{(1)}_\ell$ and~$\m^{(2)}_\ell$ converge.
Using that the Lagrangian is continuous on~$\K \times \K$, in~\eqref{Sm2} we can pass to the limit.
This proves that the action is indeed continuous in the sense~\eqref{Scont}.
\end{proof}

Now Propositions~\ref{thm:main1} and~\ref{prop:LipRepara} extend in a straightforward way.
The only additional ingredient to keep in mind is that the 
causal Lagrangian is indeed H\"older continuous with H\"older exponent~$\alpha=1/(2n+1)$
(see~\cite[Theorems~5.1 and~5.3]{banach}), so that we can use the estimate~\eqref{hoeles}.
\begin{theorem}\label{thm:cfs1}
For any~$\xi \geq 0$, there is a  H\"older continuous flow
\[ (\m^\xi, f^\xi) \in\hold^{0,\frac{1}{2}}([0,\infty);{\mathcal{P}}(\K)) \]
with~$(\m^\xi, f^\xi)(0)=(\m_0, f_0)$. Setting
\[ 
t_{\max} := \inf \big\{ t \in \R^+ \:\big|\: \Sact\big( \rho^\xi(t) \big) = \inf_{\tau \in \R^+}
\Sact \big( \rho^\xi(\tau) \big) \big\} \:, \]
the action is strictly monotone decreasing up to~$t_{\max}$, i.e.\
\[ 
\Sact\big( \m^\xi(t_1), f^\xi(t_1) \big) > \Sact\big( \m^\xi(t_2), f^\xi(t_2) \big) \qquad
\qquad \text{for all~$0 \leq t_1 < t_2 \leq t_{\max}$}\:. \]
Moreover, the flow curve satisfies for all~$0 \leq t_1 < t_2 \leq t_{\max}$ the H\"older bound
\[ 
d \Big( \big( m^{\xi}(t_{1}), f^\xi(t_1) \big), \big( \m^{\xi}(t_{2}), f^\xi(t_2) \big) \Big)\leq \sqrt{2}\: \sqrt{t_{2}-t_{1}}\;
\sqrt{\Sact\big( \m^\xi(0), f^\xi(0) \big) } \:. \]

Finally, in the case~$\xi>0$, this curve satisfies the Lipschitz bound
\begin{align*}
d &\Big( \big( \m^{\xi}(t_1), f^\xi(t_1) \big), \big( \m^{\xi}(t_2), f^\xi(t_2) \big) \Big) \\
& \leq \frac{1}{\xi} \: \big( S(\rho^\xi\big( \m^\xi(t_1), f^\xi(t_1) \big) - S(\rho^\xi\big( \m^\xi(t_2), f^\xi(t_2) \big) \big) \:.
\end{align*}
\end{theorem}

Following the procedure in Section~\ref{seclip}, in the case~$\xi>0$, we may reparametrize
using the action itself as the parameter~$s$. We denote the reparametrized curve again with
an additional tilde, i.e.\
\[ (\tilde{\m}^\xi, \tilde{f}^\xi) : \big( \Sact^\xi_{\min}, \Sact(\rho_0) \big] \rightarrow {\mathcal{P}}(\K) \:. \]
In analogy to Proposition~\ref{prop:LipRepara}, we have the following result.
\begin{proposition}\label{prop:LipReparacfs} The curve~$(\tilde{\m}^\xi, \tilde{f}^\xi)$ is
Lipschitz continuous in the sense that
\begin{align*}
d&\big( (\tilde{\m}^\xi(s_1), \tilde{f}^\xi)(s_1), (\tilde{\m}^\xi(s_2), \tilde{f}^\xi)(s_2) \big) \notag \\
&\leq \frac{1}{\xi} \: \big( s_2 - s_1 \big) \qquad
\text{for all~$\Sact^\xi_{\min} \leq s_1 < s_2 \leq \Sact(\varrho_0)$} \:. 
\end{align*}
Moreover, the limit~$(\tilde{\m}^\xi, \tilde{f})(\mathcal{S}_{\min}^{\xi}):=\mathrm{w}^{*}\text{-}\lim_{s\searrow \Sact_{\min}^{\xi}}\big(\tilde{\m}^\xi(s), \tilde{f}^\xi)(s)\big)$
exists in the sense of weak*-conver\-gence of measures.
\end{proposition}

\subsection{Limiting measures and Euler-Lagrange equations}
Theorems~\ref{thmlimit1} and~\ref{thmlimit2} extend in a straightforward way to causal fermion systems.
Since the assumptions in Theorem~\ref{thmlimit1} are strong and seem difficult to verify in the
applications, we only state the analog of Theorem~\ref{thmlimit2}.
%

\begin{theorem} \label{thmlimit2cfs} In the case~$\xi>0$, for any~$q>0$
the curve~$(\tilde{\m}^\xi(s), \tilde{f}^\xi(s))$ converges with respect to the distance function~\eqref{d2} as~$s \searrow \Sact^\xi_{\min}$.
In the case of penalization by the Wasserstein distance~$W_p$
(i.e.\ in Case~2.\ in~\eqref{dspecify}), the limiting measure
\[ (\m^\xi_\infty, f^\xi_\infty) := \lim_{s \searrow \Sact^\xi_{\min}} (\tilde{\m}^\xi(s), \tilde{f}^\xi(s)) \]
satisfies the EL equations approximately, in the sense that the function~$\ell_\xi$ defined by
\[ \ell_\xi(x) := \int_\F \L(x,y)\: \dif \rho(y) + \frac{\xi}{2} \: d \big( (\delta_{z}, \lambda),
(\m^\xi_\infty, f^\xi_\infty) \big) \]
is minimal on the support of~$\rho$, i.e.\
\[ 
\ell_\xi|_N \equiv \inf_\F \ell_\xi \]
with~$N:= \supp \rho$ and~$\rho$ defined similar to~\eqref{rholimit}
by $\rho:= \tilde{F}_* \tilde{m}^\xi$ and~$\tilde{F}(x) := \tilde{f}(x)\, x$.
\end{theorem}
\begin{proof}
We again proceed as in Section~\ref{seclimit}, always with the measures in~$\meas_1(\F)$
replaced by pairs in~${\mathcal{P}}(\K)$ (see~\eqref{PKdef}).
The existence of the limit measure follows as in Proposition~\ref{prop:LipRepara}.
The EL equation are obtained exactly as in Lemma~\ref{lem:verticaldown}.
\end{proof}
We finally point out that the last proof of convergence no longer applies if~$\xi=0$.
This is the reason why in Theorem~\ref{thmlimit1} we had to {\em{assume}} that
the curve~$(\m^0(t), f^0(t))$ converges. Similar as explained by the example in Section~\ref{secex},
in the case~$\xi=0$ we cannot expect convergence of the curve.

\section{Application and outlook: A flow in the infinite-dimensional case} \label{secoutlook}
In order to exemplify possible applications of the constructed flows, we will now show
how the Lipschitz continuous flow constructed in Proposition~\ref{prop:LipReparacfs}
can be used in order to construct a corresponding flow in the {\em{infinite-dimensional
setting}}. The general idea is to append the flows in finite-dimensional subspaces of
the Hilbert space for increasing dimension.

For the detailed construction, we 
assume that the Hilbert space~$\H$ in Definition~\ref{defcfs} is separable
but~$\dim \H=\infty$. We consider a filtration by finite-dimensional subspaces, i.e.\
\beq \label{filtration}
\H_1 \subset \H_2 \subset \cdots \subset \H \quad \text{with} \quad \dim \H_p = p \qquad \text{and} \qquad \H = 
\overline{\bigcup_{p=1}^\infty \H_p \!\!\!}^{\la .|. \ra_\H} \:.
\eeq
Extending the operators by zero, we obtain corresponding inclusions~$\F_{1}\subset\F_{2}\subset ...\subset\F$ with
\begin{align*}
\meas_1(\F_{1}) \stackrel{\iota_{1}}{\hookrightarrow} \meas_1(\F_{2}) \stackrel{\iota_{2}}{\hookrightarrow}...
\end{align*}
for suitable embedding maps~$\iota_{j}$, $j\in\mathbb{N}$. 

Given a parameter~$\xi>0$ and a starting point~$(\mathfrak{m}_0, f_0) \in {\mathcal{P}}(\K)$, we consider the reparametrized flow from Proposition~\ref{prop:LipReparacfs}
in~$\F_1$. It has a limit point, i.e.\
\[ \lim_{s \searrow \Sact^\xi_{{\min}, 1}} (\tilde{\mathfrak{m}}^\xi, \tilde{f}^\xi)(s) =
(\tilde{\mathfrak{m}}_0, \tilde{f}_0) \in {\mathcal{P}}(\K) \:. \]
Using the above embeddings, we can consider this limiting measure as being in~${\mathcal{P}}(\K)$. Taking this measure as the new starting point, we consider the
repara\-me\-trized flow from Proposition~\ref{prop:LipReparacfs} in~$\F_2$. Proceeding in this way inductively,
we obtain a Lipschitz continuous curve in~$\F$. The action is strictly decreasing along the flow curve.

We note that the above method can be refined in various ways. One extension which
seems useful is not to choose~$\xi>0$ as a constant, but to consider instead a monotone decreasing sequence~$(\xi_p)_{p \in \N}$ which converges to zero as the dimension~$p$ of the Hilbert space tends to infinity. Similarly one can also adjust the parameter~$\kappa$
in~\eqref{Lagrange} when increasing the dimension.
The detailed construction remains to be worked out.

We finally remark that this procedure is inspired by and bears some resemblance with
 {\em{renormalization flow techniques}} used in quantum
field theory. In order to explain the connection, we note that ultraviolet regularizations are 
often realized by a cutoff in momentum space which (at least for systems in finite spatial
volume) corresponds to restricting attention to finite-dimensional subspaces of
the underlying Hilbert space. Removing the cutoff corresponds to the limit when the
dimensions of the subspaces tend to infinity. In the renormalization program, one
studies this limit while carefully adjusting the masses and coupling constants
in the physical action. Our analysis is similar because we study minimizers of the
causal action for a filtration~\eqref{filtration} while adjusting the parameters~$\xi$
and~$\kappa$.

\Thanks{{{\em{Acknowledgments:}}
We would like to thank the referees for the careful reading and many
useful suggestions. F.G.\ would like to thank the Hector Foundation for support.


\bibliographystyle{amsplain}
\providecommand{\bysame}{\leavevmode\hbox to3em{\hrulefill}\thinspace}
\providecommand{\MR}{\relax\ifhmode\unskip\space\fi MR }
\providecommand{\MRhref}[2]{%
  \href{http://www.ams.org/mathscinet-getitem?mr=#1}{#2}
}
\providecommand{\href}[2]{#2}

\end{document}